\documentclass{vldb}

\vldbTitle{The PGM-index: a multicriteria, compressed and learned approach to data indexing}
\vldbAuthors{Paolo Ferragina and Giorgio Vinciguerra}
\vldbDOI{https://doi.org/10.14778/xxxxxxx.xxxxxxx}
\vldbVolume{12}
\vldbNumber{xxx}
\vldbYear{2019}

\usepackage{placeins}
\usepackage{balance}
\usepackage{subcaption}
\usepackage{booktabs}
\usepackage{multirow}
\usepackage{pgfplots}
\usepackage{pgfplotstable}
\usepackage{xspace}
\usepackage{mathtools}
\usepackage{hyperref}
\usepackage{glossaries}
\usepackage{glossaries-extra}
\usepackage{hyphenat}
\usepackage[noabbrev,capitalize,nameinlink]{cleveref}

\usepgfplotslibrary{colorbrewer}
\usepgfplotslibrary{groupplots} 
\pgfplotsset{
  compat=1.14,
}

\glssetcategoryattribute{acronym}{nohyperfirst}{true}

\setabbreviationstyle[acronym]{long-short}
\newacronym{rmi}{RMI}{Recursive Model Index}
\newacronym{ml}{ML}{Machine Learning}
\newacronym{em}{EM}{External Memory}
\newacronym{dag}{DAG}{Directed Acyclic Graph}
\newacronym{mae}{MAE}{Mean Absolute Error}
\newacronym{csstree}{CSS-tree}{Cache-Sensitive Search tree}
\newacronym{plamodel}{PLA-model}{Piecewise Linear Approximation model}

\newcommand{\pgmindex}{PGM\hyp{}index\xspace}
\newcommand{\bplustree}{\mbox{B\textsuperscript{+}-tree}\xspace}

\newcommand{\fitingtree}{FITing\hyp{}tree\xspace}
\newcommand{\fitingtrees}{FITing\hyp{}trees\xspace}
\newcommand{\btree}{\mbox{B-tree}\xspace}
\newcommand{\btrees}{\mbox{B-trees}\xspace}
\newcommand{\text@hyphens}{\mathcode`\-=`\-\relax}

\newcommand{\id}[1]{\ensuremath{\mathit{\text@hyphens#1}}}
\newcommand{\myparagraph}[1]{\medskip\noindent\textbf{#1.}\hspace{1mm}}
\DeclarePairedDelimiter{\floor}{\lfloor}{\rfloor}

\newdef{definition}{Definition}
\newtheorem{theorem}{Theorem}
\newtheorem{lemma}{Lemma}
\newtheorem{proposition}{Proposition}

\begin{document}

\title{The PGM-index: a multicriteria, compressed and learned approach to data indexing}

\numberofauthors{2}
\author{
\alignauthor
Paolo Ferragina \\
\affaddr{University of Pisa, Italy} \\
\email{paolo.ferragina@unipi.it} \\
\alignauthor
Giorgio Vinciguerra \\
\affaddr{University of Pisa, Italy} \\
\email{giorgio.vinciguerra@phd.unipi.it}
}

\maketitle

\begin{abstract}
The recent introduction of learned indexes has shaken the foundations of the decades-old field of indexing data structures. Combining, or even replacing, classic design elements such as \btree nodes with machine learning models has proven to give outstanding improvements in the space footprint and time efficiency of data systems.
However, these novel approaches are based on heuristics, thus they lack any guarantees both in their time and space requirements.

We propose the Piecewise Geometric Model index (shortly, \pgmindex), which achieves guaranteed I/O-optimality in query operations, learns an optimal number of linear models, and its peculiar recursive construction makes it a purely learned data structure, rather than a hybrid of traditional and learned indexes (such as RMI and \fitingtree). We show experimentally that the \pgmindex improves the space of the best known learned index, i.e. \fitingtree, by 63.3\% and of the \btree by more than four orders of magnitude, while achieving their same or even better query time efficiency.

We complement this result by proposing three variants of the \pgmindex which address some key issues occurring in the design of modern big data systems. First, we design a {\em compressed} \pgmindex that further reduces its succinct space footprint by exploiting the repetitiveness at the level of the learned linear models it is composed of. Second, we design a \pgmindex that adapts itself to the distribution of the query operations, thus resulting in the first known {\em distribution-aware} learned index to date. Finally, given its flexibility in the offered space-time trade-offs, we propose the {\em multicriteria} \pgmindex whose speciality is to efficiently auto-tune itself in a few seconds over hundreds of millions of keys to the possibly evolving space-time constraints imposed by the application of use. 
\end{abstract}

\section{Introduction}\label{sec:introduction}

The ever-growing amount of information coming from the Web, social networks and Internet of Things severely impairs the management of available data. Advances in CPUs, GPUs and memories hardly solve this problem without properly devised algorithmic solutions. Hence, much research has been devoted to dealing with this enormous amount of data, particularly focusing on memory hierarchy utilisation~\cite{Bender:2005,Vitter:2001}, query processing on streams~\cite{Cormode:2017}, space efficiency~\cite{Navarro:2016,Navarro:2007}, parallel and distributed processing~\cite{Grama:2003}. But despite these formidable results, we still miss proper algorithms and data structures that are flexible enough to work under computational constraints that vary across users, devices and applications, and possibly evolve over time.

In this paper, we restrict our attention to the case of {\em indexing data structures} for internal or external memory which solve the so-called {\em fully indexable dictionary} problem. This problem asks to store a multiset $S$ of real keys in order to efficiently support the query $\id{rank}(x)$, which returns for any possible key $x$ the number of keys in $S$ which are smaller than $x$. In formula,  $\id{rank}(x) = |\{y \in S \mid y < x\}|$. Now, suppose that the keys in $S$ are stored in a sorted array $A$. It is not difficult to deploy the \id{rank} primitive to implement the following classic queries:

\begin{itemize}
    \item $\id{member}(x) = \textsc{true}$ if $x \in S$, \textsc{false} otherwise. Just check whether $A[\id{rank}(x)]=x$, since $A$ stores items from position $0$.
    
    \item $\id{predecessor}(x) = \max\{y \in S \mid y < x\}$. Return $A[i]$, where $i=\id{rank}(x)-1$.

    \item $\id{range}(x,y) = S \cap [x, y]$. Scan from $A[\id{rank}(x)]$ up to keys smaller than or equal to $y$.
\end{itemize}

\noindent Moreover, we notice that it is easy to derive from $\id{member}(x)$ the implementation of the query $\id{lookup}(x)$, which returns the satellite data of $x \in S$ (if any), \textsc{nil} otherwise.

In the following, we will use the generic expression {\em query operations} to refer to any of the previous kinds of pointwise queries, namely: $\id{member}(x)$, $\id{predecessor}(x)$ and $\id{lookup}(x)$. On the contrary, we will be explicit in referring to $\id{range}(x,y)$ because of its variable-size output.

\myparagraph{Background and related work}
Existing indexing data structures can be grouped into: (i)~hash-based, which range from traditional hash tables to recent techniques, like Cuckoo hashing~\cite{Pagh:2004}; (ii)~tree-based, such as \btrees and its variants~\cite{Bender:2005,Vitter:2001,Rao:1999}; (iii)~bitmap-based~\cite{Witten:1999,Chan:1998}, which allow efficient set operations; and (iv)~trie-based, which are commonly used for string keys. Unfortunately, hash-based indexes do not support predecessor or range searches; bitmap-based indexes can be expensive to store, maintain and decompress~\cite{Wang:2017}; trie-based indexes are mostly pointer-based and, apart from recent results~\cite{Ferragina:2016}, keys are stored uncompressed thus taking space proportional to the dictionary size. As a result, \btrees and their variations remain the predominant data structures in commercial database systems for these kinds of queries~\cite{Petrov:2018}.\footnote{For other related work we refer the reader to \cite{Kraska:2018,Galakatos:2019}, here we mention only the results which are closer to our proposal.}

Very recently, this old-fashioned research field has been shaken up by the introduction of {\em learned indexes}~\cite{Kraska:2018}, whose combination, or even replacement, of classic design elements, such as B-tree nodes, with machine-learned models have been shown to achieve outstanding improvements in the space footprint and time efficiency of all the above query operations. The key idea underlying these new data structures is that indexes are {\em models} that we can train to map keys to their location in the array $A$, given by \id{rank}. This parallel between indexing data structures and \id{rank} functions does not seem a new one, in fact any of the previous four families of indexes offers a specific implementation of it. But its novelty becomes clear when we look at the keys $k\in S$ as points $(k, \id{rank}(k))$ in the Cartesian plane. As an example, let us consider the case of a dictionary of keys $a, a+1, \dots, a+n-1$, where $a$ is an integer. Here, \id{rank}$(k)$ can be computed {\em exactly} as $k-a$ (i.e. via a line of slope $1$ and intercept $-a$), and thus it takes constant time and space to be implemented, independently of the number $n$ of keys in $S$. This trivial example sheds light on the potential compression opportunities offered by patterns and trends in the data distribution. However, we cannot argue that all datasets follow exactly a ``linear trend''. 

In general, we have to design \gls{ml} techniques that learn \id{rank} by extracting the patterns in the data through succinct models, ranging from linear to more sophisticated ones, which admit some ``errors'' in the output of the model approximating \id{rank} and that, in turn, can be efficiently corrected to return the exact value of \id{rank}. This way, we can reframe the implementation of \id{rank} as a \gls{ml} problem in which we search for the model that is fast to be computed, is succinct in space, and best approximates \id{rank} according to some criteria that will be detailed below.

This is exactly the design goal pursued by~\cite{Kraska:2018} with their \gls{rmi}, which uses a hierarchy of \gls{ml} models organised as a \gls{dag} and trained to learn the input distribution $(k, \id{rank}(k))$ for all $k\in S$. At query time each model, starting from the top one, takes the query key as input and picks the following model in the \gls{dag} that is ``responsible'' for that key. The output of \gls{rmi} is the position returned by the last queried \gls{ml} model, which is, however, an approximate position. A final binary search is thus executed within a range of neighbouring positions whose size depends on the prediction error of \gls{rmi}.

One could presume that \gls{ml} models cannot provide the guarantees ensured by traditional indexes, both because they can fail to learn the distribution and because they can be expensive to evaluate \cite{Kraska:2019}. Unexpectedly, it was reported that \gls{rmi} dominates the \bplustree, being up to 1.5--3$\times$ faster and two orders of magnitude smaller in space~\cite{Kraska:2018}.

This notwithstanding, the \gls{rmi} introduces another set of space-time trade-offs between model size and query time which are difficult to control because they depend on the distribution of the input data, on its \gls{dag} structure and on the complexity of the \gls{ml} models adopted. This motivated~\cite{Galakatos:2019} to introduce the \fitingtree which uses only linear models, a \bplustree to index them, and it provides an integer parameter $\varepsilon\geq1$ controlling the size of the region in which the final binary search step has to be performed. \cref{fig:segment} shows an example of a linear model $f_s$ approximating 14 keys and its use in determining the approximate position of a key $k=37$, which is indeed $f_s(k)\approx 7$ instead of the correct position $5$, thus making an error $\varepsilon=2$. Experiments showed that the \fitingtree improves the time performance of the \bplustree with a space saving of orders of magnitude~\cite{Galakatos:2019}, but this result was not compared against the performance of \gls{rmi}. Moreover, the computation of the linear models residing in the leaves of the \fitingtree is sub-optimal in theory and inefficient in practice. This impacts negatively on its final space occupancy (as we will quantify in Section \ref{ssec:exp-space-occupancy}) and slows down its query efficiency because of an increase in the height of the \bplustree indexing those linear models.

\begin{figure}[t]
  \centering
  \begin{tikzpicture}
\newcommand{\vertLineFromPoint}[1]{
  \draw[dotted] 
  (axis cs:#1) -- ({axis cs:#1}|-{rel axis cs:0,0})
}
\newcommand{\horLineFromPoint}[1]{
  \draw[dotted] 
  (axis cs:#1) -- ({axis cs:#1}-|{rel axis cs:0,0})
}
\begin{axis}[
  height=5.7cm,
  width=\columnwidth,
  xlabel={key},
  ylabel={position},
  xlabel shift=-5pt,
  ylabel shift=-5pt,
  tick align=center,
  tick pos=left,
  enlargelimits=0.06,
  ymin=0,
  ymax=13,
  x grid style={white!69.019607843137251!black},
  y grid style={white!69.019607843137251!black},
  ytick={0,1,2,3,4,5,8,9,10,11,12,13},
  extra x ticks={37},
  extra x tick labels={$k$},
  extra y ticks={6.623},
  extra y tick labels={$f_s(k)$},
]
\addplot [black, mark=*, mark size=1, mark options={solid}, only marks, forget plot]
table {%
27 0
29 1
32 2
32 3
33 4
37 5
37 6
37 7
38 8
40 9
41 10
43 11
44 12
46 13
};
\addplot [no markers, densely dashed, semithick, gray] table [y={create col/linear regression={y=Y}}] {%
X Y
27 0
29 1
32 2
33 4
37 5
38 8
40 9
41 10
43 11
44 12
46 13
};
\vertLineFromPoint{37,37*\pgfplotstableregressiona+\pgfplotstableregressionb};
\horLineFromPoint{37,37*\pgfplotstableregressiona+\pgfplotstableregressionb};
\horLineFromPoint{37,5};
\draw[decorate,decoration=brace] (axis cs:19.5,37*\pgfplotstableregressiona+\pgfplotstableregressionb-3) -- (axis cs:19.5,37*\pgfplotstableregressiona+\pgfplotstableregressionb+3) node [midway,rotate=90,yshift=1.2em] {\footnotesize$[\id{pos}-\varepsilon, \id{pos}+\varepsilon]$};
\end{axis}
\end{tikzpicture}
  \vspace{-2em}
  \caption{Linear approximation of a multiset of integer keys within the range $[27,46]$. The encoding of the (dashed) segment $f_s$ takes only two floats, and thus its storage is independent of the number of ``encoded'' keys. The key $k=37$ is repeated three times in the multiset $S$, starting from position $5$ in $A$, but $f_s$ errs by $\varepsilon=2$ in predicting the position of its first occurrence.}
  \label{fig:segment}
\end{figure}
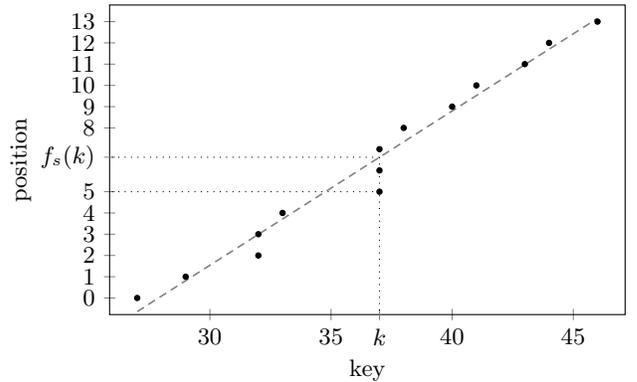

\myparagraph{Our contribution} 
In this paper, we contribute to the design of {\em optimal} linear-model learned indexes, to their compression and to the automatic selection of the best learned index that fits the requirements (in space, latency or query distribution) of an underlying application in five main steps.
\begin{enumerate}
  \item We design the first learned index that solves the fully indexable dictionary problem with time and space complexities which are provably better than classic data structures for hierarchical memories, such as \btrees, and modern learned indexes. Our index is I/O-optimal according to the lower bound for predecessor search in external memory proved by \cite{Patrascu:2006}. We call it the \emph{Piecewise Geometric Model index} (\pgmindex) because it turns the indexing of a sequence of keys into the coverage of a sequence of points via segments. Unlike previous work~\cite{Galakatos:2019,Kraska:2018}, the \pgmindex is built upon an optimal number of linear models, and its peculiar recursive construction makes it a purely learned data structure, rather than hybrid of traditional and learned data structures. This aspect allows the \pgmindex to make the most of the constant space-time indexing feature offered by the linear models on which it is built upon (see \cref{ssec:indexing-pla-model,thm:pgm-index,tab:query-costs}).
  
  \item We test the experimental efficiency of the \pgmindex through a large set of experiments over three known datasets (\cref{sec:experiments}). We show that the \pgmindex\ improves the space occupancy of the \fitingtree by 63.3\%, of the \gls{csstree} by a factor 82.7$\times$, and of the \btree by more than four orders of magnitude, while achieving their same or even better query efficiency (see \cref{fig:space-time-comparison}). Unlike the \gls{rmi}, the \pgmindex offers theoretical bounds on the query time and space occupancy, and it guarantees a 4$\times$ increase in the precision of approximating the position of the searched key which, in turn, induces a uniform improvement over all possible space-time trade-offs achieved by \gls{rmi}. 

  \item We then show that the (succinct) space footprint of a \pgmindex can be further reduced by designing novel compression algorithms for the building blocks of the linear models (i.e. slopes and intercepts) on which our index hinges upon. In particular, we provide an efficient algorithm that reduces the number of distinct slopes to be encoded to their optimal minimum number, which is an interesting algorithmic contribution in itself. In practice, in just 80 ms this algorithm improves the space occupancy of a \pgmindex over a dataset of hundreds of million keys by up to 52.2\%. This makes the \pgmindex the first compressed learned index to date (see \cref{sec:compressed-pgm-index}).

  \item We also propose the first example of a distribution-aware learned index, namely one that adapts itself not only to the distribution of the dictionary keys but also to their access frequencies. The resulting distribution-aware \pgmindex achieves the query time of biased data structures~\cite{Bagchi:2005,Bent:1985,Mehlhorn:1984,Seidel:1996}, but with a space occupancy that adapts to the ``regularity trend'' of the input dataset thus benefiting of the succinctness of learned indexes (see \cref{sec:distribution-aware-pgm-index,thm:pgm-index-distribution-aware}).  
  
  \item Given the flexibility in space-time trade-offs offered by the \pgmindex (as shown in \cref{sec:experiments}), we finally study the concept of {\em Multicriteria Data Structures}, which combines in a formal yet effective way multicriteria optimisation with data structures. A multicriteria data structure, for a given problem $P$, is defined by a pair $\langle \mathcal F, \mathcal A \rangle_P$ where $\mathcal F$ is a family of data structures, each one solving $P$ with a proper trade-off in the use of some resources (e.g. time, space, energy), and $\mathcal A$ is an optimisation algorithm that selects in $\mathcal F$ the data structure that ``best fits'' an instance of $P$. We demonstrate the fruitfulness of this concept by introducing the Multicriteria \pgmindex, which hinges upon an optimisation algorithm designed to efficiently explore $\mathcal F$ via a proper space-time cost model for our \pgmindex (\cref{sec:multicriteria-pgm-index}). In our experiments, we show that the Multicriteria \pgmindex is fast, taking less than 20 seconds, to reorganise itself to {\em best} index a dataset of 750M keys within newly given space or time bound. This supports the vision of a new generation of big data processing systems designed upon data structures that can be adjusted on-the-fly to the application, device and user needs, which may possibly change over time, as foreseen in~\cite{Idreos:2019}. In a way, the multicriteria \pgmindex solves their ambitious research challenge within a design space in which $\mathcal F$ consists of variants of \pgmindex, space-time constraints change continuously, and the data structure has to be optimised as fast as possible.\footnote{For completeness, we remark that the concept of multicriteria optimisation has been already applied in Algorithmics to data compression~\cite{Farruggia:2014} and software auto-tuning~\cite{Naono:2010}.}
\end{enumerate}

\section{The \pgmindex}

Given a multiset $S$ of $n$ keys drawn from a universe $\mathcal{U}$,\footnote{The universe $\mathcal{U}$ is a range of reals because of the arithmetic operations required by the linear models. Our solution works for any kind of keys that can be mapped to reals by preserving their order. Examples include integer keys, string keys, etc.} the \pgmindex is a data structure parametric in an integer $\varepsilon \geq 1$ which solves the fully indexable dictionary problem introduced in \cref{sec:introduction}. Let $A$ be a sorted array storing the (possibly repeated) keys of $S$. 

The first ingredient of the \pgmindex is a \glsreset{plamodel}\gls{plamodel}, namely a mapping between keys from $\mathcal U$ and their approximate positions in the array $A$. Specifically, we aim to learn a mapping that returns a position for a key $k \in \mathcal U$ which is at most $\varepsilon$ away from the correct one in $A$. We say {\em piecewise} because one single linear model (i.e. a segment) could be insufficient to $\varepsilon$-approximate the positions of all the keys from $\mathcal U$. As a consequence, the \pgmindex learns a {\em sequence} of segments, each one taking constant space (two floats and one key) and constant query time to return the $\varepsilon$-approximate position of $k$ in $A$. We show below in \cref{lem:optimal-pla-model} that there exists a linear time and space algorithm which computes the optimal \gls{plamodel}, namely one that consists of the minimum number of $\varepsilon$-approximate segments. We also observe that the $\varepsilon$-approximate positions returned by the optimal \gls{plamodel} can be turned into {\em exact} positions via a binary search within a range of $\pm\varepsilon$ keys in $A$, thus taking time logarithmic in the parameter $\varepsilon$, not in the size of $A$. 

The second ingredient of the \pgmindex is a recursive algorithm which adapts the index structure to the distribution of the input keys, thus resulting as much independent as possible from their number (see \cref{fig:pgm-index} for pictorial example). More precisely, in order to make the most of the ability of a single segment to index in constant space and time an arbitrarily long range of keys, we turn the optimal \gls{plamodel} built over the array $A$ into a set of keys, and we proceed recursively by building another optimal \gls{plamodel} over these keys. This process continues until one single segment is obtained, which will form the root of our data structure.
\begin{figure*}[ht] 
  \centering
  \resizebox{0.72\textwidth}{!}{\input{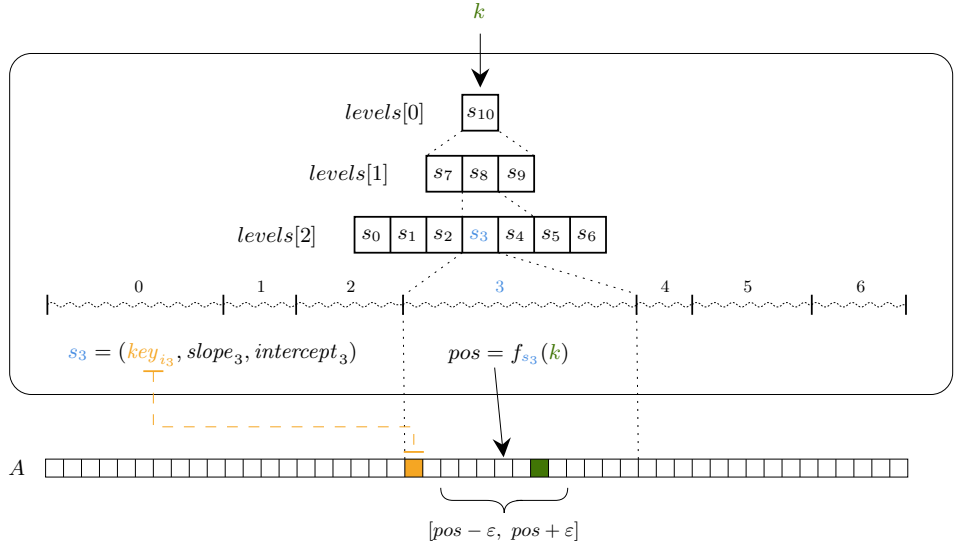}}
  \caption{Each segment in a \pgmindex is ``responsible'' for routing the queried key to one of the segments of the level below. In the picture, the dotted lines show that the root segment $s_{10}$ routes the queried key to one of the segments among $\{s_7, s_8, s_9\}$ (which together cover the same range of keys as $s_{10}$), whereas $s_8$ routes the queried key to either $s_3$ or $s_4$ (which together cover the same range of keys as $s_8$). Segments at the last level (i.e. $\id{levels}[2]$) are $\varepsilon$-approximate segments for the sub-range of keys in $A$ depicted by wavy lines of which they are responsible for.}
  \label{fig:pgm-index}
\end{figure*}
Overall, each \gls{plamodel} forms a level of the \pgmindex, and each segment of that \gls{plamodel} forms a node of the data structure at that level.
The speciality of this recursive construction with respect to known learned index proposals (cf. \fitingtree or \gls{rmi}) is that the \pgmindex is a {\em pure learned index} which does not hinge on classic data structures either in its structure (as in the \fitingtree) or as a fallback when the \gls{ml} models err too much (as in \gls{rmi}). 

The net result are three main advantages in its space-time complexity. First, the \pgmindex is built upon the minimum number of segments, while other learned indexes, such as \fitingtree and \gls{rmi}, compute a sub-optimal number of segments with a subsequent penalisation in their time and space efficiency.
Second, the \pgmindex uses these segments as constant-space routing tables at all levels of the data structure, while other indexes (e.g. \fitingtree, \btree and variants) use space-consuming nodes storing a large number of keys which depends on the disk-page size only, thus resulting blind to the possible regularity present in the data distribution. Third, these routing tables of the \pgmindex take constant time to restrict the search of a key in a node to a smaller subset of the indexed keys (of size $\varepsilon$), whereas nodes in the \bplustree and the \fitingtree incur a search cost that grows with the node size, thus slowing the tree traversal during the query operations.

\smallskip The following two subsections will detail the two main ingredients of the \pgmindex described above.

\subsection{The optimal PLA-model}\label{ssec:piecewise-linear-model}

Let us be given a sorted array $A=[k_0, k_1, \dots, k_{n-1}]$ of $n$ real and possibly repeated keys drawn from a universe $\mathcal{U}$. In this section, we describe how an $\varepsilon$-approximate implementation of the mapping \id{rank} from keys to positions in $A$ can be efficiently computed and succinctly stored via an optimal number of segments, which is one of the core design elements of a \pgmindex. In the next section, we will comment on the recursive construction of the whole \pgmindex and the implementation of the query operations.

A segment $s$ is a triple $(k, \id{slope}, \id{intercept})$ that indexes a range of $\mathcal{U}$ through the function $f_{s}(k) = k\, \times\, \id{slope} + \id{intercept}$, as depicted in \cref{fig:segment}. An important characteristic of the \pgmindex is the ``precision'' $\varepsilon$ of its segments.

\begin{definition}%[$\varepsilon$-approximate segment]
  Let $A$ be a sorted array of $n$ keys drawn from a universe $\mathcal U$ and let $\varepsilon \geq 1$ be an integer. A segment $s=(k,\id{slope},\id{intercept})$ is said to provide an \emph{$\varepsilon$-approximate} indexing of the range of all keys in $[k_i,k_{i+r}]$, for some $k_i,k_{i+r} \in A$, if $|f_{s}(x)-\id{rank}(x)| \leq \varepsilon$ for all $x\in \mathcal{U}$ such that $k_i \leq x \leq k_{i+r}$.
\end{definition}

An $\varepsilon$-approximate segment can be seen as an approximate predecessor search data structure for its covered range of keys offering constant query time and constant occupied space. One single segment, however, could be insufficient to $\varepsilon$-approximate the $\id{rank}$ function over the whole $\mathcal U$; hence, we look at the computation of a sequence of segments, also termed \glsreset{plamodel}\gls{plamodel}. 

\begin{definition}%[Optimal \gls{plamodel}]
Given $\varepsilon \geq 1$, the \emph{piecewise linear $\varepsilon$\hyp{}approximation problem} consists of computing the \gls{plamodel} which minimises the number of its segments $\{s_0, \ldots, s_{m-1}\}$, provided that each segment $s_j$ is $\varepsilon$-approximate for its covered range of keys, these ranges are disjoint and together cover the entire universe $\mathcal U$.
\end{definition}

A way to find the optimal \gls{plamodel} for an array $A$ is by dynamic programming, but the $O(n^3)$ time it requires is prohibitive. The authors of the \fitingtree \cite{Galakatos:2019} attacked this problem via a heuristic approach, called shrinking cone, which is linear in time but does not guarantee to find the optimal \gls{plamodel}, and indeed it performs poorly in practice (as we will show in \cref{ssec:exp-space-occupancy}).

Interestingly enough, we found that this problem has been extensively studied for lossy compression and similarity search of time series (see e.g. \cite{ORourke:1981,Buragohain:2007,Chen:2013,Chen:2007,Xie:2014} and refs therein), and it admits streaming algorithms which take $O(n)$ optimal time. The key idea of this family of approaches is to reduce the piecewise linear $\varepsilon$-approximation problem to the one of constructing convex hulls of a set of points, which in our case is the set $\{(k_i,\id{rank}(k_i))\}$ grown incrementally for $i=0,\dots,n-1$. As long as the convex hull can be enclosed in a (possibly rotated) rectangle of height no more than $2\varepsilon$, the index $i$ is incremented and the set is extended. As soon as the rectangle enclosing the convex hull is higher than $2\varepsilon$, we compute one segment of the \gls{plamodel} by taking the line which splits that rectangle into two equal-sized halves. Then, the current set of processed elements is emptied and the algorithm restarts from the rest of the input points. This greedy approach can be proved to be optimal and to have linear time and space complexity. We can rephrase this result in our context as follows.

\begin{lemma}[Optimal \gls{plamodel} \cite{ORourke:1981}]
  \label{lem:optimal-pla-model}
  Given a sequence $\{(x_i,y_i)\}_{i=0,\dots,n-1}$ of points that are nondecreasing in their $x$-coordinate. There exists a streaming algorithm that in linear time and space computes the minimum number of segments that $\varepsilon$-approximate the $y$-coordinate of each point in that set.
\end{lemma}

For our application to the dictionary problem, the $x_i$s of \cref{lem:optimal-pla-model} correspond to the input keys $k_i$, and the $y_i$s correspond to their positions $0, \dots, n-1$ in the sorted input array $A$. 
Therefore, \cref{lem:optimal-pla-model} provides an algorithm which computes in linear time and space the optimal \gls{plamodel} for the keys stored in $A$.

The next step is to prove a simple but very useful bound on the number of keys covered by a segment of the optimal \gls{plamodel}, which we will deploy in the analysis of the \pgmindex.

\begin{lemma}
  \label[lemma]{lem:lower-bound-2eps}
  Given an ordered sequence of keys $k_i \in \mathcal{U}$ and the corresponding sequence $\{(k_i,i)\}_{i=0,\dots,n-1}$ of points in the Cartesian plane that are nondecreasing in both their coordinates. The  algorithm of \cref{lem:optimal-pla-model} determines a (minimum) number $m_\textit{opt}$ of segments which cover at least $2\varepsilon$ points each, so that $m_\textit{opt} \leq n/(2\varepsilon)$.
\end{lemma}
\begin{proof}
  For any chunk of $2\varepsilon$ consecutive keys $k_i, k_{i+1},\allowbreak\dots,\allowbreak k_{i+2\varepsilon-1}$, let us take the horizontal segment $y = i+\varepsilon$. It is easy to see that those keys generate the points $(k_i, i), (k_{i+1}, i+1), \dots , (k_{i+2\varepsilon-1}, i+2\varepsilon-1)$ and each of these keys have $y$-distance at most $\varepsilon$ from that line, which is then an $\varepsilon$\hyp{}approximate segment for that range of $2\varepsilon$-keys. Hence, any segment of the optimal \gls{plamodel} covers at least $2\varepsilon$ keys.
\end{proof}

\subsection{Indexing the PLA-model}\label{ssec:indexing-pla-model}

The algorithm of \cref{lem:optimal-pla-model} returns an optimal \gls{plamodel} for the input array $A$ as a sequence $M=[s_0,\dots,s_{m-1}]$ of $m$ segments.\footnote{To simplify the notation, we write $m$ instead of $m_\textit{opt}$.} Now, in order to solve the fully indexable dictionary problem, we need a way to find the $\varepsilon$-approximate segment $s_j=(k_{i_j},\id{slope}_j,\id{intercept}_j)$ responsible for estimating the approximate position \id{pos} of a query key $k$, i.e. this is the rightmost segment $s_j$ such that $k_{i_j} \leq k$.
When $m$ is large, we could perform a binary search on the sequence $M$, or we could index it via a proper data structure, such as a multiway search tree (as done in the \fitingtree). In this case, the membership query could then be answered in three steps. First, the multiway search tree is queried to find the rightmost segment $s_j$ such that $k_{i_j} \leq k$. Second, that segment $s_j$ is used to estimate the position $\id{pos} = f_{s_j}(k)$ for the query key $k$. Third, the exact position of $k$ is determined via a binary search within $A[\id{pos}-\varepsilon, \id{pos}+\varepsilon]$. The net consequence is that a query over this data structure would take $O(\log_B m + \log\varepsilon)$ time, where $B$ is the  fan-out of the multiway tree and $\varepsilon$ is the error incurred by $s_j$ when approximating $\id{rank}(k)$.

However, the indexing strategy above does not take full advantage of the key distribution because it resorts to a classic data structure with fixed fan-out to index $M$. Therefore, we introduce a novel strategy which consists of repeating the piecewise linear approximation process recursively on a set of keys derived from the sequence of segments. More precisely, we start with the sequence  $M$ constructed over the whole input array $A$, then we extract the first key of $A$ covered by each segment and finally construct another optimal \gls{plamodel} over this reduced set of keys. We proceed in this recursive way until the \gls{plamodel} consists of one segment. 

If we map segments to nodes, then this approach constructs a sort of multiway search tree but with three main advantages with respect to \btrees (and thus \fitingtrees): (i)~its nodes have variable fan-out driven by the (typically large) number of keys covered by the segments associated with those nodes; (ii)~the segment in a node plays the role of a constant-space and constant-time $\varepsilon$-approximate routing table for the various queries to be supported; (iii)~the search in each node corrects the $\varepsilon$-approximate position returned by that routing table via a binary search (see next), and thus it has a time cost that depends logarithmically on $\varepsilon$, and hence it is independent of the number of keys covered by the corresponding segment.

Now, a query operation over this \emph{Recursive \pgmindex} works as follows. At every level, it uses the segment referring to the visited node to estimate the position of the searched key $k$ among the keys of the lower level.\footnote{To correctly approximate the position of a key $k$ falling between the last key covered by a segment $s_j$ and the first key covered by $s_{j+1}$, we compute $\min\{f_{s_j}(k), f_{s_{j+1}}(k_{i_{j+1}})\}$.} The real position is then found by a binary search in a range of size $2\varepsilon$ centred around the estimated position. Given that every key on the next level is the first key covered by a node on that level, we have identified the next node to visit, and thus the next segment to query, and the process continues until the last level is reached. An example of a query operation is depicted in \cref{fig:pgm-index}.

\begin{theorem}\label{thm:pgm-index}
  Let $A$ be an ordered array of $n$ keys from a universe $\mathcal{U}$, and $\varepsilon\geq 1$ be a fixed integer parameter. The Recursive \pgmindex with parameter $\varepsilon$ indexes the array $A$ taking $\Theta(m)$ space and answers rank, membership and predecessor queries in $O(\log m)$ time and $O((\log_c m) \log(\varepsilon/B))$ I/Os, where $m$ is the minimum number of $\varepsilon$-approximate segments covering $A$, $c \geq 2\varepsilon$ denotes the variable fan-out of the data structure, and $B$ is the block size of the External Memory model. Range queries are answered in extra (optimal) $O(\id{K})$ time and $O(K/B)$ I/Os, where $K$ is the number of keys satisfying the range query.
\end{theorem}
\begin{proof}
Each step of the recursion reduces the number of segments by a variable factor $c$ which is nonetheless at least $2\varepsilon$ because of \cref{lem:lower-bound-2eps}. The number of levels is, therefore, $L=O(\log_c m)$, and the total space required by the index is $\sum_{\ell=0}^L m/(2\varepsilon)^\ell=\Theta(m)$. For the rank, membership and predecessor queries, the bounds on the running time and the I/O complexity follow easily by observing that a query performs $L$ binary searches over intervals having size at most $2\varepsilon$. In the case of range queries, we report the $K$ keys by scanning $A$ from the position returned by the rank query.
\end{proof}

\begin{table*}[ht]
  \newcolumntype{C}[1]{>{\centering\let\newline\\\arraybackslash\hspace{0pt}}m{#1}}%
  \centering
  \begin{tabular}{m{4cm} C{2cm} C{3cm} C{3cm} C{3cm}}
  \toprule
      Data structure
    & Space
    & RAM model \newline \footnotesize{worst case time}
    & EM model \newline \footnotesize{worst case I/Os}
    & EM model \newline \footnotesize{best case I/Os} \\
  \midrule
  Plain sorted array                 & $O(1)$                  & $O(\log n)$                 & $O(\log \frac{n}{B})$ & $O(\log \frac{n}{B})$ \\
  Multiway tree (e.g. \btree)       & $\Theta(n)$             & $O(\log n)$                 & $O(\log_B n)$ & $O(\log_B n)$ \\
  \fitingtree  & $\Theta(m_\textit{greedy})$ & $O(\log m_\textit{greedy}+\log\varepsilon)$     & $O(\log_B m_\textit{greedy})$         & $O(\log_B m_\textit{greedy})$ \\
  Recursive \pgmindex                & $\Theta(m_\textit{opt})$    & $O(\log m_\textit{opt}+\log\varepsilon)$               & $O(\log_c m_\textit{opt})$ \newline \scriptsize{$c\geq2\varepsilon=\Omega(B)$} & $O(1)$ \\
  \bottomrule%\\[0.005em]
  \end{tabular}
  \caption{The Recursive \pgmindex improves the time, I/O and space complexity of the query operations of traditional external memory indexes (e.g. \btree) and learned indexes (i.e. \fitingtree). The integer $\varepsilon\geq 1$ denotes the error we guarantee in approximating the positions of the input keys. We denote with $m_\textit{opt}$ the minimum number of $\varepsilon$-approximate segments, computed by \cref{lem:optimal-pla-model}, and with $m_\textit{greedy}$ the number of $\varepsilon$-approximate segments computed by the greedy algorithm at the core of the \fitingtree. Of course, $m_\textit{opt} \leq m_\textit{greedy}$. The learned index \gls{rmi} is not included in the table because it lacks of guaranteed bounds.}
\label{tab:query-costs}
\end{table*}

The main novelty of the \pgmindex is that its space overhead does not grow linearly with $n$, as in the traditional indexes mentioned in \cref{sec:introduction}, but it depends on the ``regularity trend'' of the input array $A$ which also decreases with the value of $\varepsilon$. Because of \cref{lem:lower-bound-2eps}, the number of segments at the last level of a \pgmindex cannot be more than $n/(2\varepsilon)$, so that $m < n$ since $\varepsilon \geq 1$. Since this fact holds for the recursive levels too, the \pgmindex cannot be asymptotically worse in space and time than a $2\varepsilon$-way tree, such as a \fitingtree, \bplustree or \gls{csstree} (just take $c=2\varepsilon=\Theta(B)$ in \cref{thm:pgm-index}). According to the lower bound proved by \cite{Patrascu:2006}, we can state that the \pgmindex solves I/O-optimally the fully indexable dictionary problem with predecessor search, meaning that it can potentially replace any existing index with virtually no performance degradation.

\cref{tab:query-costs} summarises these bounds for the \pgmindex and its competitors both in the RAM model and in the \gls{em} model for the rank query and its derivatives: i.e. predecessor, membership and lookup.  

The thorough experimental results of \cref{sec:experiments} will support further these theoretical achievements by showing  that the \pgmindex is much faster and succinct than \fitingtree, \bplustree and \gls{csstree} because, in practice, it will be $m_{opt} \ll n$ and $c \gg 2\varepsilon$.

\section{The Compressed \pgmindex}\label{sec:compressed-pgm-index}

Compressing the \pgmindex boils down to providing proper lossless compressors for the keys and the segments (i.e. intercepts and slopes) which constitute the building blocks of our learned data structure.
In this section, we propose techniques specifically tailored to the compression of the segments, since the compression of the input keys is an orthogonal problem for which there exist a plethora of solutions (see e.g. \cite{Moffat:2002book,Navarro:2016} for integer keys, and \cite{Burtscher:2009,Lindstrom:2018} for floating-point keys).

For what concerns the compression of the intercepts, they can be made increasing by using the coordinate system of the segments, i.e. the one that computes the position of an element $k$ as $f_{s_j}(k)=(k-k_{i_j}) \times \id{slope}_j+\id{intercept}_j$. Then, since the result of $f_{s_j}(k)$ has to be truncated to return an integer position in $A$, we store the intercepts as integers $\floor{\id{intercept}_j}$.\footnote{Note that this transformation increases $\varepsilon$ by 1.} Finally, we exploit the fact that intercepts are increasing integers smaller than $n$ and thus use the succinct data structure of \cite{Okanohara:2007} to obtain the following result.

\begin{proposition}
  Let $m$ be the number of segments of a \pgmindex indexing $n$ keys drawn from a universe $\mathcal{U}$. The intercepts of those segments can be stored using $m\log(n/m)+1.92m+o(m)$ bits and be randomly accessed in $O(1)$ time. 
\end{proposition}

The compression of slopes is more involved, and we need to design a specific novel compression technique. The starting observation is that the algorithm of \cref{lem:optimal-pla-model} computes not just a single segment but a whole family of $\varepsilon$-approximate segments whose slopes identify to an interval of reals.
Specifically, let us suppose that these slope intervals are $I_0=(a_0,b_0), \dots, I_{m-1}=(a_{m-1},b_{m-1})$, where each original slope $\id{slope}_j$ belongs to $I_j$ for $j=0, \, \ldots, m-1$. Our goal is to reduce the entropy of the set of slopes by reducing their distinct number from $m$ to $t$. This will allow us to change the encoding of $m$ floats into the encoding $t$ floats plus $m$ short integers, with the hope that $t \ll m$ as we will indeed show experimentally in \cref{ssec:exp-space-occupancy}.

We achieve this goal by designing an algorithm that takes $O(m \log m)$ time to determine the {\em minimum number $t$ of distinct slopes} such that each interval $I_j$ includes one of them. Given this result, we create a table $T$ which stores the distinct slopes (as $t$ floating-point numbers of $64$ bits each) and then change every initial $\id{slope}_j \in I_j$ into one of the new $t$ distinct slopes, say $\id{slope}'_j$, which is still guaranteed to belong to $I_j$ and can be encoded in $\lceil \log t \rceil$ bits.

Let us now describe the algorithm to determine the minimum number $t$ of distinct slopes for the initial slope intervals which are assumed to be $(a_0,b_0), \dots, (a_{m-1},b_{m-1})$. First, we sort lexicographically the slope intervals to obtain an array $I$ in which overlapping intervals are consecutive. We assume that every pair keeps as satellite information the index of the corresponding interval, namely $j$ for $(a_j,b_j)$. Then, we scan $I$ to determine the maximal prefix of intervals in $I$ that intersect each other. As an example, say the sorted slope intervals are $\{(2,7), (3,6),\allowbreak (4,8), (7,9),\dots\}$, then the first maximal sequence of intersecting intervals is $\{(2,7), (3,6), (4,8) \}$ because they intersect each other but the fourth interval $(7,9)$ does not intersect the second interval $(3,6)$ and thus it is not included. 

Let $(l,r)$ be the intersection of all the intervals in the current maximal prefix of $I$: this is $(4,6)$ in the running example. Then, any slope in $(l,r)$ is an $\varepsilon$-approximate slope for everyone of the intervals in that prefix of $I$. Therefore, we choose one real in $(l,r)$ and assign it as the slope of each of those segments in that maximal prefix. The process then continues by determining the maximal prefix of the remaining intervals, until the overall sequence $I$ is processed (details and optimally proof in the full version of the paper).

\begin{lemma}
  \label{lem:slope-compression-algorithm}
  Let $m$ be the number of $\varepsilon$-approximate segments of a \pgmindex indexing $n$ keys drawn from a universe $\mathcal{U}$. There exists a lossless compressor for the segments which computes the minimum number of distinct slopes $t \leq m$ while preserving the $\varepsilon$-guarantee. The algorithm takes $O(m \log m)$ time and compresses the slopes into $64t+m\lceil \log t \rceil$ bits of space. 
\end{lemma}
\begin{proof}
The compressed space occupancy of the $t$ distinct slopes in $T$ is, assuming double-precision floats, $64t$ bits. The new slopes $\id{slope}'_j$ are still $m$ in their overall number, but each of them can be encoded as the position $0, \dots, t-1$ into $T$ of its corresponding double-precision float.
\end{proof}

An interesting future work is to experiment how much the use of universal coders for reals \cite{Lindstrom:2018}, as an alternative to floating-point numbers, can further reduce the additive term $64 t$.

\section{The Distribution-Aware \pgmindex}\label{sec:distribution-aware-pgm-index}

The \pgmindex of \cref{thm:pgm-index} implicitly assumes that the queries are uniformly distributed, but this seldom happens in practice. For example, queries in search engines are very well known to follow skewed distributions such as the Zipf's law~\cite{Witten:1999}. In such cases, it is desirable to have an index that answers the most frequent queries faster than the rare ones, so to achieve a higher query throughput. Previous work exploited query distribution in the design of binary trees~\cite{Bent:1985}, Treaps~\cite{Seidel:1996}, and skip lists~\cite{Bagchi:2005}, to mention a few.

In this section, we introduce an orthogonal approach that builds upon the \pgmindex by proposing a variant that adapts itself not only to the distribution of the input keys but also to the distribution of the queries. This turns out to be the {\em first distribution-aware learned index} to date, with the additional positive feature of being very succinct in space. 

Formally speaking, given a sequence $S=\{(k_i,p_i)\}_{i=1,\dots,n}$, where $p_i$ is the probability to query the key $k_i$ (that is assumed to be known), we want to solve the distribution-aware dictionary problem, which asks for a data structure that searches for a key $k_i$ in time $O(\log (1/p_i))$ so that the average query time coincides with the entropy of the query distribution $\mathcal H=\sum_{i=1,\dots,n} p_i \log (1/p_i)$.

We note that the algorithm of \cref{lem:optimal-pla-model} can be modified so that, given a $y$-range for each one of $n$ points in the plane, finds also the set of all (segment) directions that intersect those ranges in $O(n)$ time (see \cite{ORourke:1981}). This corresponds to find the optimal \gls{plamodel} whose individual segments guarantee an approximation which is within the $y$-range given for each of those points. Therefore, our key idea is to define, for every key $k_i$, a $y$-range of size $y_i = \min{\{\/1/p_i, \varepsilon\}}$, and then apply the algorithm of \cref{lem:optimal-pla-model} on that set of keys and $y$-ranges. Clearly, for the keys whose $y$-range is $\varepsilon$ we can use \cref{thm:pgm-index} and derive the same space bound of $O(m)$; whereas for the keys whose $y$-range is $1/p_i < \varepsilon$ we observe that these keys are no more than $\varepsilon$ (in fact, the $p_i$s sum up to $1$), but they are possible spread among all position in $A$ and thus they induce in the worst case $2\varepsilon$ extra segments. Therefore, the total space occupancy of the bottom level of the index is $\Theta(m+\varepsilon)$, where $m$ is the one defined in \cref{thm:pgm-index}. Now, let us assume that the search for a key $k_i$ arrived at the last level of this Distribution-Aware \pgmindex, and thus we know in which segment to search for $k_i$: the final binary search step within the $\varepsilon$-approximate range returned by that segment takes $O(\log \min\{1/p_i, \varepsilon\}) = O(\log (1/p_i))$ as we aimed for.

We are left with showing how to find that segment in a distribution-aware manner. We proceed similarly to the Recursive \pgmindex but with a careful design of the recursive step because of the probabilities (and thus the variable $y$-ranges) assigned to the recursively defined set of keys.

Let us consider the segment covering the sequence $S_{[a,b]}=\{(k_a,p_a), \dots, (k_b, p_b)\}$, denote by $q_{a,b}=\max_{i\in[a,b]}p_i$ the maximum probability of a key in $S_{[a,b]}$, and by  $P_{a,b} = \sum_{i=a}^b p_i$ the cumulative probability of all keys in $S_{[a,b]}$ (which is indeed the probability to end up in that segment when searching for one of its keys). We create the new set $S'=\{\dots,\allowbreak (k_a,q_{a,b}/P_{a,b}), \dots\}$ formed by the first key $k_a$ covered by each segment (as in the recursive \pgmindex) and setting its associated probability to $q_{a,b}/P_{a,b}$. Then, we construct the next upper level of the Distribution-Aware \pgmindex by applying the algorithm of \cref{lem:optimal-pla-model} on $S'$. If we iterate the above analysis for this new level of weighted segments, we conclude that: if we know from the search executed on the levels above that $k_i \in S_{[a,b]}$, the time cost to search for $k_i$ in this level is $O(\log\min\{P_{a,b}/q_{a,b}, \varepsilon\}) = O(\log (P_{a,b}/p_i))$.

Let us repeat this argument for another upper level in order to understand the influence on the search time complexity. We denote the segment that covers the range of keys which include $k_i$ with $S_{[a',b']} \supset S_{[a,b]}$, the cumulative probability with $P_{a',b'}$, and thus assign to its first key $k_{a'}$ the probability $r/P_{a',b'}$, where $r$ is the maximum probability of the form $P_{a,b}$ of the ranges included in $[a',b']$. In other words, if $[a',b']$ is partitioned into $\{z_1,\dots,z_c\}$, then $r=\max_{i\in[1,c)} P_{z_i,z_{i+1}}$. Reasoning as done previously, if we know from the search executed on the levels above that $k_i \in S_{[a',b']}$, the time cost to search for $k_i$ in this level is $O(\log \min{\{\/P_{a',b'}/r, \varepsilon\}}) = O(\log (P_{a',b'}/P_{a,b}))$ because $[a,b]$ is, by definition, one of these ranges in which $[a',b']$ is partitioned. 
 
Repeating this design until one single segment is obtained (whose cumulative probability is one), we get a total time cost for the search in all levels of the \pgmindex equal to a sum of logarithms whose arguments ``cancel out'' and get $O(\log (1/p_i))$.  

As far as the space bound is concerned, we recall that the number of levels in the \pgmindex is $L = O(\log_c m)$ with $c \geq 2 \varepsilon$, and that we have to account for the $\varepsilon$ extra segments per level returned by the algorithm of \cref{lem:optimal-pla-model}. Consequently, this distribution-aware variant of the \pgmindex takes $O(m + L \varepsilon)$ space. This space bound is indeed $O(m)$ because $\varepsilon$ is a constant parameter (see \cref{sec:experiments}).
\begin{theorem}
  \label{thm:pgm-index-distribution-aware}
  Let $A$ be an ordered array of $n$ keys $k_i$ drawn from a universe $\mathcal{U}$, which are queried with (known) probability $p_i$, and let $\varepsilon\geq 1$ be a fixed integer parameter. The Distribution-Aware Recursive \pgmindex with parameter $\varepsilon$ indexes the array $A$ in $O(m)$ space and answers queries in $O(\mathcal H)$ average time, where $\mathcal H$ is the entropy of the query distribution, and $m$ is the number of segments of the optimal \gls{plamodel} for the keys in $A$ with error $\varepsilon$.
\end{theorem}

\newpage
\section{The Multicriteria \pgmindex}
\label{sec:multicriteria-pgm-index}
Tuning a data structure to match the application's needs is often a difficult and error-prone task for a software engineer, not to mention that these needs may change over time due to mutations in data distribution, devices, resource requirements, and so on. The typical approach is a grid search on the various instances of the data structure to be tuned until the one that matches the application's needs is found. However, the data structure may be not flexible enough to adapt to those changes, or the search space can be so huge that the reorganisation of the data structure takes too much time.

In the rest of this section, we exploit the space-time flexibility of the \pgmindex by showing that this tuning process can be efficiently automated over this data structure via an optimisation strategy that: (i) given a space constraint outputs the \pgmindex that minimises its query time; or symmetrically, (ii) given a maximum query time outputs the \pgmindex that minimises its space footprint.

\myparagraph{The time-minimisation problem}
According to \cref{thm:pgm-index}, the query time of a Recursive \pgmindex can be described as $t(\varepsilon) = c (\log_{2\varepsilon} m) \log(2\varepsilon/B)$, where $B$ is the page size of the \gls{em} model, $m$ is the number of segments in the last level, and $c$ depends on the access latency of the memory. For the space, we introduce $s_\ell(\varepsilon)$, which denotes the number of segments needed to have precision $\varepsilon$ over the keys available at level $\ell$ of the Recursive \pgmindex, and compute the overall number of segments as $s(\varepsilon) = \sum_{\ell=1}^L s_\ell(\varepsilon)$. By \cref{lem:lower-bound-2eps}, we know that  $s_L(\varepsilon) = m \le n / (2\varepsilon)$ for any $\varepsilon \geq 1$ and that $s_{\ell-1}(\varepsilon) \leq s_\ell(\varepsilon)/(2\varepsilon)$. So that $s(\varepsilon) \leq \sum_{\ell=0}^L m/(2\varepsilon)^\ell = (2\varepsilon m-1)/(2\varepsilon-1)$.

Given a space bound $s_{max}$, the ``time-minimisation problem'' consists of minimising $t(\varepsilon)$ subject to $s(\varepsilon) \leq s_{max}$.\footnote{For simplicity, we assume that a disk page contains exactly $B$ keys. This assumption can be relaxed by putting the proper machine- and application-dependent constants in front of $t(\varepsilon)$ and $s(\varepsilon)$.} The greatest challenge here is that we do not have a closed formula for $s(\varepsilon)$, but only an upper bound which does not depend on the underlying dataset as $s(\varepsilon)$ does. \cref{sec:experiments} will show that in practice we can model $m=s_L(\varepsilon)$ with a simple power-law having the form $a \varepsilon^{-b}$, whose parameters $a$ and $b$ will be properly estimated on the dataset at hand. The power-law covers both the pessimistic case of \cref{lem:lower-bound-2eps} and the best case in which the dataset is strictly linear.

Clearly, the space occupancy decreases with increasing $\varepsilon$, whereas the query time $t(\varepsilon)$ increases with $\varepsilon$, since the number of keys on which it is executed a binary search at each level equals $2\varepsilon$. Therefore, the time-minimisation problem reduces to the problem of finding the value of $\varepsilon$ for which $s(\varepsilon)=s_{max}$ because it is the lowest $\varepsilon$ that we can afford. Such value of $\varepsilon$ could be found by a binary search in the bounded interval $\mathcal E = [B/2, n/2]$ which is derived by requiring that each model errs at least a page size (i.e. $2\varepsilon \geq B$), since lower $\varepsilon$ values do not save I/Os, and by observing that one model is the minimum possible space (i.e.  $2\varepsilon \leq n$, by \cref{lem:lower-bound-2eps}). However, provided that our power-law approximation holds, we can speed up the search of that "optimal" $\varepsilon$ by guessing the next value of $\varepsilon$ rather than taking the midpoint of the current search interval. In fact, we can find the root of $s(\varepsilon) - s_{max}$, i.e. the value $\varepsilon_g$ for which $s(\varepsilon_g) = s_{max}$. We emphasise that such $\varepsilon_g$ may not be the solution of our problem, as it may be the case that the approximation or the fitting of $s(\varepsilon)$ by means of a power-law is not precise. Thus, more iterations of the search may be needed to find the optimum $\varepsilon$ value; anyway, we guarantee to be always faster than a binary search by gradually switching to it. Precisely, we bias the guess $\varepsilon_g$ towards the midpoint $\varepsilon_m$ of the current search range via a simple convex combination of the two~\cite{Graefe:2006}.

\myparagraph{The space-minimisation problem}
Given a time bound $t_{max}$, the space-minimisation problem consists of minimising $s(\varepsilon)$ subject to $t(\varepsilon) \leq t_{max}$. As for the problem above, we could perform a binary search inside the interval $\mathcal E = [B/2, n/2]$ and look for the maximum $\varepsilon$ which satisfies that time constraint. Instead, we speed up this process by guessing the next iterate for $\varepsilon$ via the equation $t(\varepsilon) =t_{max}$, that is solving $c (\log_{2\varepsilon} s_L(\varepsilon)) \log(2\varepsilon/B)=t_{max}$, in which $s_L(\varepsilon)$ is replaced by the power-law approximation $a\varepsilon^{-b}$, for proper $a$ and $b$, and $c$ is replaced by the measured memory latency of the given machine. 

Although effective, this approach raises a subtle issue, namely, the time model could not be a correct estimate of the actual query time because of hardware-dependent factors such as the presence of several caches and the CPU pre-fetching. To further complicate this issue, we note that both $s(\varepsilon)$ and $t(\varepsilon)$ depend on the power-law approximation $a\varepsilon^{-b}$.

For these reasons, instead of using the time model $t(\varepsilon)$ to steer the search, we measure and use the actual average query time $\overline t(\varepsilon)$ of the \pgmindex over a fixed batch of random queries. Moreover, instead of performing a binary search inside the whole $\mathcal E$, we run an exponential search starting from the solution of the dominating term $c\log(2\varepsilon/B)=t_{max}$, i.e. the cost of searching the data. Eventually, we stop the search of the best $\varepsilon$ as soon as the searched range is smaller than a given threshold because $\overline t(\varepsilon)$ is subject to measurement errors (e.g. due to an unpredictable CPU scheduler).

\section{Experiments}\label{sec:experiments}

We experimented with an implementation in C\texttt{++} of the \pgmindex on a machine with a 2.3 GHz Intel Xeon Gold and 192 GiB memory.\footnote{The implementation will be released on GitHub with the acceptance of this paper.} We used the following three standard datasets, each having different data distributions, regularities and patterns: 

\begin{itemize}

\item \emph{Web logs}~\cite{Kraska:2018,Galakatos:2019} contains timestamps of about 715M requests to a web server; 
\item \emph{Longitude}~\cite{OpenStreetMap} contains longitudes of about 166M points-of-interest from OpenStreetMap; 

\item \emph{IoT}~\cite{Kraska:2018,Galakatos:2019} contains timestamps of about 26M events recorded by IoT sensors installed throughout an academic building. 

\end{itemize}

We also generated some synthetic datasets according to the uniform distribution in the interval $[0,u)$, to the Zipf distribution with exponent $s$, and to the lognormal distribution with standard deviation $\sigma$, they will allow to test thoroughly the various indexes.

\subsection{Space occupancy of the \pgmindex}
\label{ssec:exp-space-occupancy}

In this set of experiments, we estimated the size of the optimal \gls{plamodel} (see \cref{ssec:piecewise-linear-model}) returned by our implementation of~\cite{Xie:2014}, which provides the segments stored in the bottom level of the \pgmindex, and compared it against the non-optimal \gls{plamodel} computed with the greedy shrinking cone algorithm~\cite{Galakatos:2019,Sklansky:1980} used in the \fitingtree~\cite{Galakatos:2019}. This comparison is important because the size of a \gls{plamodel} is the main factor impacting the space footprint of a learned index based on linear models.

\cref{tab:synthetic} shows that on synthetic datasets of $10^9$ keys the improvements (i.e. relative change in the number of segments) ranged from 20.8\% to 69.4\%. \cref{fig:vs-fitting-tree} confirms these trends also for real-world datasets, on which the improvements ranged from 37.7\% to 63.3\%. For completeness, we report that the optimal algorithm with $\varepsilon=8$ built a \gls{plamodel} for Web logs in 2.59 seconds, whereas it took less than 1 second for \emph{Longitude} and \emph{IoT} datasets. This means that the optimal algorithm of \cref{ssec:piecewise-linear-model} can scale to even larger datasets.

\begin{table}[tb]
  \centering
  \begin{tabular}{l r r r r r}
    \toprule
    \multirow{2}{1cm}{Dataset} & \multicolumn{5}{c}{$\varepsilon$} \\
    \cmidrule{2-6}
                           &    8 &   32 &  128 &  512 & 2048 \\
    \midrule
    Uniform $u=2^{22}$     & 33.8 & 59.4 & 66.5 & 68.6 & 68.8 \\
    Uniform $u=2^{32}$     & 65.8 & 68.3 & 68.9 & 69.4 & 68.7 \\
    Zipf $s=1$             & 47.9 & 59.0 & 62.8 & 44.7 & 29.0 \\
    Zipf $s=2$             & 45.3 & 40.2 & 24.2 & 20.8 & 21.6 \\
    Lognormal $\sigma=0.5$ & 66.1 & 68.5 & 68.8 & 62.1 & 35.6 \\
    Lognormal $\sigma=1.0$ & 66.1 & 68.4 & 69.0 & 61.9 & 34.5 \\
    \bottomrule
  \end{tabular} %\bigskip
  \caption{Space savings of \pgmindex with respect to a \fitingtree for a varying $\varepsilon$ on six synthetic datasets of 1G elements generated according to the specified distributions. The \pgmindex saved from 20.8\% to 69.4\%.}
  \label{tab:synthetic}
\end{table}

\begin{figure}[t]
  \begin{tikzpicture}
\pgfplotsset{
  colormap={printerfriendly}{rgb255=(67,162,202) rgb255=(168,221,181) rgb255=(224,243,219) },
  cycle list/.define={printerfriendly}{[of colormap=printerfriendly]},
}
\pgfplotsset{
    /pgfplots/layers/mylayers/.define layer set={
        axis background,axis grid,axis ticks,main,axis lines,axis tick labels,
        axis descriptions,axis foreground
    }{/pgfplots/layers/standard},
}
\begin{axis}[
  width=1.05\columnwidth,
  ybar=2pt,
  bar width=4pt,
  height=6.5cm,
  enlarge x limits=0.08,
  xmode=log,
  xtick=data,
  log basis x=2,
  tick align=center,
  tick pos=left,
  minor y tick num=4,
  xlabel=$\varepsilon$,
  ylabel={Space saving (\%)},
  %ylabel shift=-7pt,
  grid=both,
  grid style={dotted},
  legend entries={{Web logs},{Longitude},{IoT}},
  legend cell align={left},
  legend style={
    at={(0.987,0.98)},
    anchor=north east,
    draw=white!80.0!black,
    /tikz/every even column/.append style={column sep=3mm},
    nodes={scale=0.7, transform shape}
  },
  mark options={solid},
  every axis plot/.append style={fill,draw=black},
  cycle list name=printerfriendly,
  legend image code/.code={
    \draw [fill=#1,draw=black] (0cm,-0.05cm) rectangle (0.4cm,0.08cm);
  },
  preaction={draw=black!95, line width=0.75pt},
  set layers=mylayers,
]

\addplot
table {%
8                            40.1
16                            45.7
32                            49.8
64                            52.7
128                            54.3
256                            54.3
512                            53.7
1024                            51.0
2048                            46.1
};

\addplot
table {%
8                            59.3
16                            63.3
32                            62.2
64                            56.2
128                            49.9
256                            45.8
512                            43.4
1024                            40.6
2048                            38.4
};

\addplot 
table {%
8                            46.1
16                            48.5
32                            50.7
64                            47.9
128                            46.8
256                            41.0
512                            37.7
1024                            44.0
2048                            43.4
};

\end{axis}
\end{tikzpicture}
  \vspace*{-0.5\baselineskip}
  \caption{The \pgmindex saved from 37.7\% to 63.3\% space with respect to a \fitingtree over the three real-world datasets. The construction time complexity of the two approaches is the same in theory (i.e. linear in the number of processed keys) and in practice (a couple of seconds, up to hundreds of millions of keys).}
  \label{fig:vs-fitting-tree}
\end{figure}
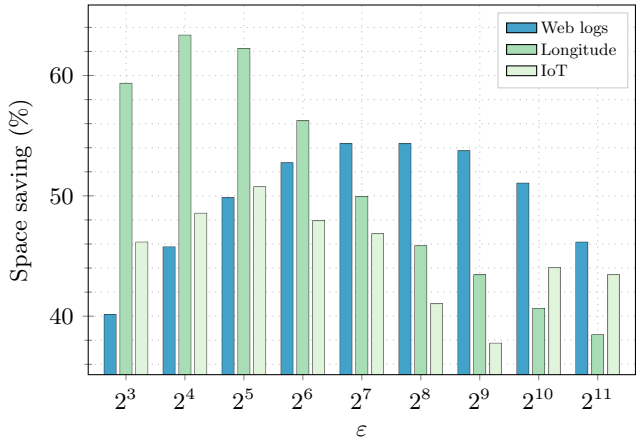

Since it appears difficult to prove a mathematical relationship between the number of input keys and the number of $\varepsilon$-approximate segments (other than the rather loose bound we proved in \cref{lem:lower-bound-2eps}), we pursued an empirical investigation on this relation because it quantifies the space improvement of learned indexes with respect to classic indexes. \cref{fig:number-of-segments} shows that, even when $\varepsilon$ is as little as $8$, the number $m$ of segments is at least two orders of magnitude smaller than the original datasets size $n$. This reduction gets impressively evident for larger values of $\varepsilon$, reaching five orders of magnitude. 
\begin{figure}[!t]
  \begin{tikzpicture}
\begin{axis}[
  width=\columnwidth,
  height=6.5cm,
  xmin=8, xmax=2048,
  xmode=log,
  xtick=data,
  ymode=log,
  log basis x=2,
  tick align=outside,
  tick pos=left,
  xlabel=$\varepsilon$,
  ylabel= {Ratio $m/n$},
  grid=major,
  grid style={dotted},
  legend entries={{Web logs},{Longitude},{IoT}},
  legend cell align={left},
  legend style={at={(0.987,0.98)},draw=white!80.0!black,nodes={scale=0.7, transform shape}},
  mark options={solid},
]

\addplot [semithick,black,mark=x,mark size=2.6]
table [y expr=\thisrowno{1}/714943653] {%
8     3909877
16    2035443
32    963148
64    423360
128   177894
256   74366
512   31608
1024  14307
2048  7204
};

\addplot [semithick,black,solid,mark=o,mark size=2.6]
table [y expr=\thisrowno{1}/166234253] {%
8     479569
16    154814
32    50248
64    19268
128   8831
256   4483
512   2425
1024  1389
2048  832
};

\addplot [semithick,black,solid,mark=triangle,mark size=2.5]
table [y expr=\thisrowno{1}/26281412]  {%
8     150482
16    73419
32    33806
64    16429
128   8253
256   4585
512   2769
1024  1493
2048  810
};

\addplot [black,dashed,mark size=2.6,domain=8:2048] {1/(2*x)};

\end{axis}
\end{tikzpicture}
  \vspace{-1\baselineskip}
  \caption{A log-log plot with the ratio between the number of segments $m$, stored in the last level of a \pgmindex, and the size $n$ of the real-world datasets as a function of $\varepsilon$. For comparison, the plot shows with a dashed line the function $1/(2\varepsilon)$ which is the fraction of the number of keys stored in the level above the input data of \bplustree with $B=2\varepsilon$ (see text). Note that $m$ is 2--5 orders of magnitude less than $n$.}
  \label{fig:number-of-segments}
\end{figure}
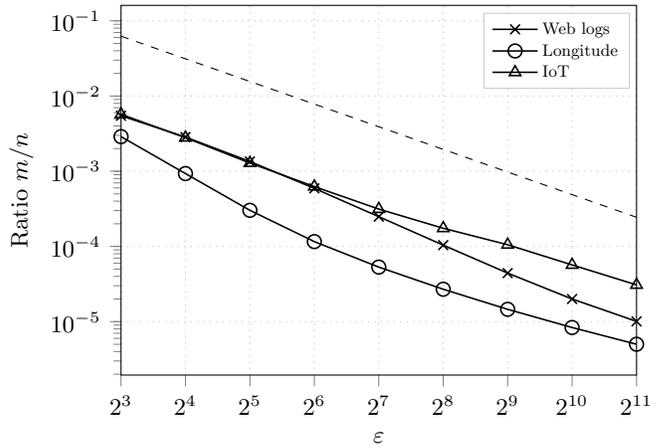

\subsection{Query performance of the \pgmindex}\label{ssec:exp-query-performance}

We evaluated the query performance of the \pgmindex and other indexing data structures on \emph{Web logs} dataset, the biggest and most complex dataset available to us. We have dropped the comparison against the \fitingtree, because of the evident structural superiority of the Recursive \pgmindex and its indexing of the optimal (minimum) number of segments in the bottom level (see \cref{fig:number-of-segments}). Nonetheless, we will investigate the performance of some variants of the \pgmindex that will provide a clear picture of the improvements determined by its recursive indexing, compared to the classic approaches based on multiway search trees (\`a la \fitingtree), \gls{csstree}~\cite{Rao:1999} or \bplustree. 

In this experiment, the dataset was loaded in memory as a contiguous array of integers represented with 8 bytes and with 128 bytes payload. Slopes and intercepts were stored as double-precision floats. Each index was presented with 10M queries randomly generated on the fly. The next three paragraphs present, respectively, the query performance of the three indexing strategies for the \pgmindex, a comparison between the \pgmindex and traditional indexes, and a comparison between the \pgmindex and the \gls{rmi}~\cite{Kraska:2018} under this experimental scenario.

\myparagraph{\pgmindex variants}
The three indexing strategies experimented for the \pgmindex are binary search, multiway tree (specifically, we implemented the \gls{csstree}~\cite{Rao:1999}) and our novel recursive construction (see \cref{ssec:indexing-pla-model}). We refer to them with PGM$\circ$BIN, PGM$\circ$CSS and PGM$\circ$REC, respectively. We set $\varepsilon_\ell=4$ for all but the last level of PGM$\circ$REC, that is the one that includes the segments built over the input dataset. Likewise, the node size of the \gls{csstree} was set to $B=2\varepsilon_\ell$ for a fair comparison with  PGM$\circ$REC. \cref{fig:space-time-variants} shows that PGM$\circ$REC dominates PGM$\circ$CSS for $\varepsilon \leq 256$, and has better query performance than PGM$\circ$BIN. The advantage of PGM$\circ$REC over PGM$\circ$CSS is also evident in terms of index height since the former has five levels whereas the latter has seven levels, thus PGM$\circ$REC experiences a shorter traversal time which is induced by a higher branching factor (as conjectured in \cref{ssec:indexing-pla-model}). For $\varepsilon > 256$ all the three strategies behaved similarly because the index was so small to fit into the L2 cache. 

\begin{figure}[!t]
  \begin{tikzpicture}
\pgfplotsset{
  every axis plot post/.append style={
    every mark/.append style={line width=1pt}
  },
  discard if not/.style 2 args={
    x filter/.append code={
      \edef\tempa{\thisrow{#1}}\edef\tempb{#2}
      \ifx\tempa\tempb\else\def\pgfmathresult{inf}\fi
    }
  }
}
\begin{axis}[
  %scale only axis,
  width=\columnwidth,
  height=6.5cm,
  xlabel={Time (ns)},
  ylabel={Index space (MiB)},
  enlarge x limits=0.08,
  enlarge y limits=0.08,
  tick align=outside,
  tick pos=left,
  grid=both,
  grid style={dotted},
  minor tick num=1,
  scaled y ticks=false,
  ytick distance={2097152},
  yticklabel = {\pgfkeys{/pgf/fpu}\pgfmathparse{\tick/1024/1024  }\pgfmathprintnumber[precision=1]{\pgfmathresult}},
  legend columns=1,
  legend cell align={left},
  legend entries={
    PGM$\circ$BIN,
    PGM$\circ$CSS,
    PGM$\circ$REC,
    $\varepsilon=64$,
    $\varepsilon=128$,
    $\varepsilon=256$,
    $\varepsilon=512$,
    $\varepsilon=1024$,
    $\varepsilon=2048$,
  },
  legend style={
    at={(0.987,0.98)},
    anchor=north east,
    draw=white!80.0!black,
    /tikz/every even column/.append style={column sep=3mm},
    nodes={scale=0.7, transform shape}
  },
]
\addlegendimage{mark size=2.5, mark=x, only marks, black}
\addlegendimage{mark size=2.5, mark=triangle, only marks, black}
\addlegendimage{mark size=2.5, mark=pentagon, only marks, black}
\pgfplotsinvokeforeach{1,...,6} {
  \addlegendimage{mark=*, only marks, index of colormap={#1 of colormap/PuBuGn-8}}
}

\pgfplotsforeachungrouped \strategy/\m in {
    BinarySearchStrategy/x,
    CSSTreeStrategy64/triangle,
    RecursiveStrategy4/pentagon}{
  \edef\temp{\noexpand\addplot [
    mark size=2.5,
    mark=\m,
    scatter,
    colormap/PuBuGn-8,
    colormap access=piecewise const,
    scatter src=explicit,
    only marks,
    discard if not={strategy}{\strategy},
  ] table [
    x=time,
    y=bytes,
    meta=logerror,
    col sep=comma,
  ] {gfx/data-xeon-gold.csv};
}\temp}
\end{axis}
\end{tikzpicture}
  %\vspace{-0.5\baselineskip}
  \caption{The query performance of several configurations of the \pgmindex. The Recursive \pgmindex, depicted as a pentagon, had better space and time performance than all the other configurations.}
  \label{fig:space-time-variants}
\end{figure}
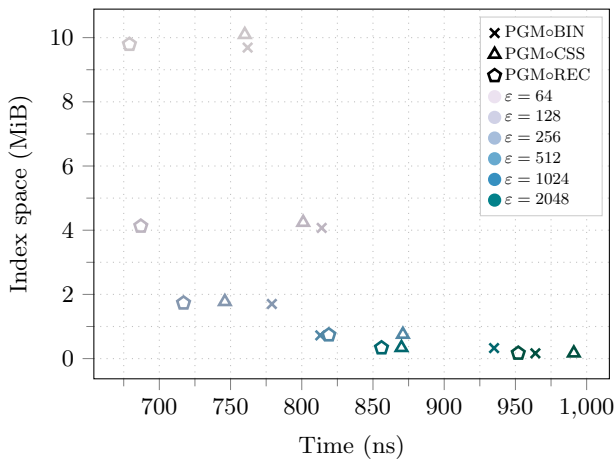

\begin{figure}[!t]
  \begin{tikzpicture}
\tikzset{
  every pin/.style={scale=0.7},
}
\pgfplotsset{
  every axis plot post/.append style={
    every mark/.append style={line width=1pt},
  },
  discard if not/.style 2 args={
    x filter/.append code={
      \edef\tempa{\thisrow{#1}}\edef\tempb{#2}
      \ifx\tempa\tempb\else\def\pgfmathresult{inf}\fi
    }
  }
}

\begin{axis}[
  %scale only axis,
  width=\columnwidth,
  height=6.5cm,
  name=ax1,
  xmax=1830,
  enlarge x limits=0.08,
  enlarge y limits=0.08,
  xlabel={Time (ns)},
  ylabel={Index space (MiB)},
  tick align=outside,
  tick pos=left,
  grid=both,
  grid style={dotted},
  minor tick num=1,
  scaled y ticks=false,
  xtick distance={200},
  ytick distance={2097152},
  yticklabel = {\pgfkeys{/pgf/fpu}\pgfmathparse{\tick/1024/1024}\pgfmathprintnumber[precision=1]{\pgfmathresult}},
  legend columns=1,
  legend cell align={left},
  ymax=14400000,
  legend entries={
    PGM $\varepsilon=2^6$,
    PGM $\varepsilon=2^7$,
    PGM $\varepsilon=2^8$,
    PGM $\varepsilon=2^9$,
    PGM $\varepsilon=2^{10}$,
    PGM $\varepsilon=2^{11}$,
    RMI 10K,
    RMI 50K,
    RMI 100K,
    RMI 200K,
    RMI 400K,
  },
  legend style={
    at={(0.987,0.98)},
    anchor=north east,
    draw=white!80.0!black,
    nodes={scale=0.7, transform shape}
  },
]

\pgfplotsinvokeforeach{1,...,6} {
  \addlegendimage{mark=pentagon, only marks, index of colormap={#1 of colormap/PuBuGn-8}}
}
\pgfplotsinvokeforeach{2,3,5,7,9} {
  \addlegendimage{mark size=2.6, mark=square, only marks, index of colormap={#1 of colormap/RdPu-9}}
}

\addplot [
  scatter,
  mark size=2.5,
  mark=pentagon,
  only marks,
  colormap/PuBuGn-8,
  point meta rel=per plot,
  point meta=\thisrow{logerror},
  discard if not={strategy}{RecursiveStrategy4}]
table [x=time, y=bytes, col sep=comma] {gfx/data-xeon-gold.csv};

\addplot [
  scatter,
  only marks,
  colormap={mine}{indices of colormap=(2,3,5,7,9 of colormap/RdPu-9)},
  mark=square,
  mark size=2,
  point meta rel=per plot,
  scatter/use mapped color={draw=mapped color,fill=none}]
table [x=nanoseconds,y=stage_size,y expr={\thisrowno{0}*(2*8+2*4)},col sep=comma] {
stage_size,error_avg,error_std,nanoseconds
10000,2702,11224,1060
50000,601,2465,853
100000,326,1100,806
200000,185,595,748
400000,106,362,805
};

\addplot [scatter, only marks, colormap={mine}{indices of colormap=(8,5,2 of colormap/YlGn-9)}, mark=diamond, mark size=2.5, point meta rel=per plot,scatter/use mapped color={mapped color}]
table {
907 11173888
1474 5586944
1763 2801664
}
node[pos=0, pin=-45:{CSS-tree 4KiB}]{}
node[pos=0.5, pin=225:{CSS-tree 8KiB}]{}
node[pos=1, pin=225:{CSS-tree 16KiB}]{};

\addplot [only marks, draw=orange, mark=star, mark size=2.5,]
table {
1193 13987056
}
node[pin=below:{B\textsuperscript{+}-tree 16KiB}]{};;
\end{axis}

\end{tikzpicture}
  %\vspace{-0.5\baselineskip}
  \caption{The Recursive \pgmindex improved uniformly \gls{rmi} with different second-stage sizes and traditional indexes with different page sizes over all possible space-time trade-offs. Results of traditional indexes for smaller page sizes are not shown because too far from the plot range. For example, the fastest \gls{csstree} occupied 341 MiB and was matched in performance by a \pgmindex of only 4 MiB (82.7$\times$ less space); the fastest \bplustree occupied 874 MiB and was matched in performance by a \pgmindex which occupied only 87 KiB (four orders of magnitude less space).}
  \label{fig:space-time-comparison}
\end{figure}
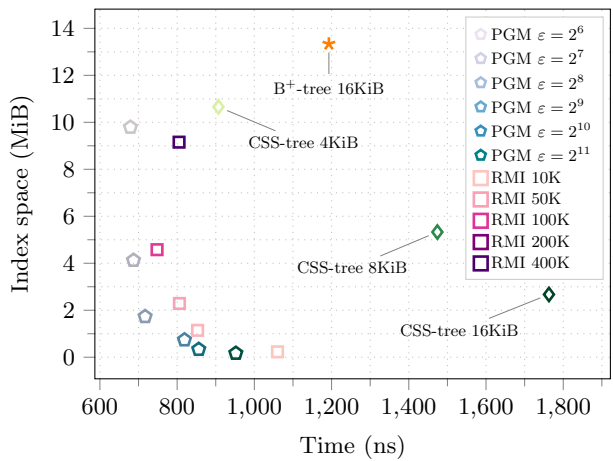

\myparagraph{\pgmindex vs traditional indexes}
We compared the \pgmindex against the cache-efficient \gls{csstree} and the classic \bplustree. For the former, we used our implementation. For the latter, we chose a well-known library~\cite{Bingmann:2013,Galakatos:2019,Kraska:2018}. 

The \pgmindex dominated these traditional indexes, as shown in  \cref{fig:space-time-comparison} for page sizes of 4--16 KiB. Performances for smaller page sizes were too far (i.e. worse) from the main plot range, and thus are not shown. For example, the fastest \gls{csstree} in our machine had page size of 128 bytes, occupied 341 MiB and was matched in query performance by a PGM$\circ$REC with $\varepsilon=128$ which occupied only 4 MiB (82.7$\times$ less space). As another example, the fastest \bplustree had page size of 256 bytes, occupied 874 MiB and was matched in query performance by a PGM$\circ$REC with $\varepsilon=4096$ which occupied only 87 KiB (four orders of magnitude less space). 

What is surprising in those plots is the improvement in space occupancy achieved by the \pgmindex which is four orders of magnitude with respect to the \bplustree and two orders of magnitude with respect to the \gls{csstree}. As stated in \cref{sec:introduction}, traditional indexes are blind to the data distribution, and they miss the compression opportunities which data trends offer. On the contrary, adapting to the data distribution through linear approximations allows the \pgmindex to uncover previously unknown space-time trade-offs, as demonstrated in this experiment.
For completeness, we report that on the 90.6 GiB of key-payload pairs the fastest \gls{csstree} took 1.2 seconds to construct, whereas the \pgmindex matching its performance in 82.7$\times$ less space took only 1.9 seconds more (despite using a single-threaded and non\hyp{}optimised computation of the \gls{plamodel}).

\myparagraph{\pgmindex vs known learned indexes}
\cref{fig:space-time-variants,fig:vs-fitting-tree} have shown that the \pgmindex improves the \fitingtree (see also the discussion at the beginning of this section). Here, we complete the comparison against the other known learned index, i.e. the 2-stage \gls{rmi} which uses a combination of linear and other models in its two stages. \cref{fig:space-time-comparison} shows that the \pgmindex dominates \gls{rmi}, it has indeed better latency guarantees because, instead of fixing the structure beforehand and inspecting the errors afterwards, it is dynamically and optimally adapted to the input data distribution while guaranteeing the desired $\varepsilon$-approximation and using the least possible space. The most compelling evidence is the \gls{mae} between the approximated and the predicted position, e.g., the \pgmindex with $\varepsilon=512$ needed about 32K segments and had \gls{mae} 226$\pm$139, while an \gls{rmi} with the same number of second stage models (i.e. number of models at the last level) had \gls{mae} 892$\pm$3729 (3.9$\times$ more). This means that \gls{rmi} experienced a higher and less predictable latency in the query execution. We report that \gls{rmi} took 30.4 seconds to construct, whereas the \pgmindex took only 3.1 seconds.

\myparagraph{Discussion} Overall, the experiments have shown that the \pgmindex is fast in construction (about 3 seconds to index a real-world table of 91 GiB with 715M key-value pairs) and has space footprint that is up to 63.3\% lower than what was achieved by a state-of-the-art \fitingtree.  Moreover, the \pgmindex dominated in space and time both the traditional and other learned index structures (e.g. the \gls{rmi}). In particular, it improved the space footprint of the \gls{csstree} by a factor 82.7$\times$ and the one of the \btree by more than four orders of magnitude, while achieving the same or even better query efficiency. 

\subsection{The Compressed \pgmindex}

We investigated the effectiveness of the compression techniques proposed in \cref{sec:compressed-pgm-index}. \cref{fig:slope-compression} shows that the slope compression algorithm reduced the {\em number} of distinct slopes significantly, up to 99.94\%, still preserving the same optimal number of segments. As far as the space occupancy is concerned, and considering just the last level of a \pgmindex which is the largest one, the reduction induced by the compression algorithm was up to 81.2\%, as shown in \cref{fig:slope-space-savings}. Note that in the \emph{Longitude} datasets for $\varepsilon \geq 2^9$ the slope compression is not effective enough. As a result, the mapping from segments to the slopes table causes an overhead that exceeds the original space occupancy of the segments. Clearly, a real-world application would turn off slope compression in such situations.

\begin{figure}[t]
  \begin{tikzpicture}

\pgfplotsset{
  colormap={printerfriendly}{rgb255=(67,162,202) rgb255=(168,221,181) rgb255=(224,243,219) },
  cycle list/.define={printerfriendly}{[of colormap=printerfriendly]},
}

\pgfplotsset{
    /pgfplots/layers/mylayers/.define layer set={
        axis background,axis grid,axis ticks,main,axis lines,axis tick labels,
        axis descriptions,axis foreground
    }{/pgfplots/layers/standard},
}

\begin{axis}[
  width=\columnwidth,
  ybar=2pt,
  bar width=4pt,
  height=6.5cm,
  enlarge x limits=0.08,
  ymin=0,
  xmode=log,
  xtick=data,
  log basis x=2,
  tick align=center,
  tick pos=left,
  minor y tick num=4,
  xlabel=$\varepsilon$,
  ylabel={Reduction in \# slopes (\%)},
  grid=both,
  grid style={dotted},
  legend entries={{Web logs},{Longitude},{IoT}},
  legend cell align={left},
  legend style={
    at={(0.987,0.98)},
    anchor=north east,
    draw=white!80.0!black,
    /tikz/every even column/.append style={column sep=3mm},
    nodes={scale=0.7, transform shape}
  },
  mark options={solid},
  every axis plot/.append style={fill,draw=black},
  cycle list name=printerfriendly,
  legend image code/.code={
    \draw [fill=#1,draw=black] (0cm,-0.05cm) rectangle (0.4cm,0.08cm);
  },
  preaction={draw=black!95, line width=0.75pt},
  set layers=mylayers,
]

\addplot
table {%
8		99.943
16		99.838
32		99.520
64		98.703
128		96.914
256		93.169
512     85.675
1024	72.559
2048	56.621
};

\addplot
table {%
8		99.234
16		96.526
32		89.034
64		74.107
128		52.587
256		29.868
512		12.825
1024	5.544
2048	1.322
};

\addplot 
table {%
8		98.414
16		96.673
32		93.229
64		87.589
128		80.722
256		74.351
512		68.003
1024	61.018
2048	54.444
};

\end{axis}
\end{tikzpicture}
  \vspace*{-0.5\baselineskip}
  \caption{The slope compression algorithm of \cref{lem:slope-compression-algorithm} reduces the number of distinct slopes by up to 99.9\%.}
  \label{fig:slope-compression}
\end{figure}
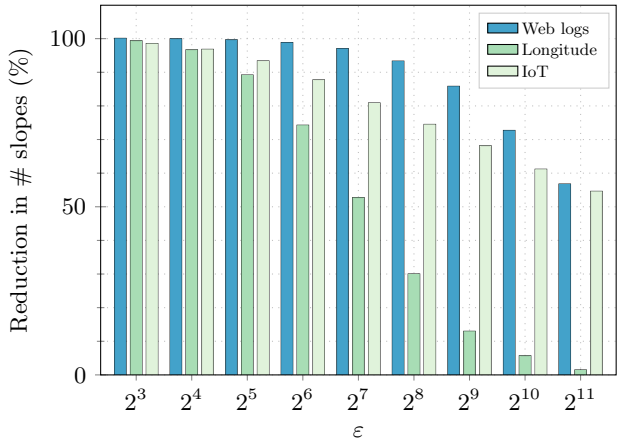

\begin{figure}[t]
  \begin{tikzpicture}

\pgfplotsset{
  colormap={printerfriendly}{rgb255=(67,162,202) rgb255=(168,221,181) rgb255=(224,243,219) },
  cycle list/.define={printerfriendly}{[of colormap=printerfriendly]},
}

\pgfplotsset{
    /pgfplots/layers/mylayers/.define layer set={
        axis background,axis grid,axis ticks,main,axis lines,axis tick labels,
        axis descriptions,axis foreground
    }{/pgfplots/layers/standard},
}

\begin{axis}[
  width=\columnwidth,
  ybar=2pt,
  ymax=110,
  bar width=4pt,
  height=6.5cm,
  enlarge x limits=0.08,
  xmode=log,
  xtick=data,
  log basis x=2,
  tick align=center,
  tick pos=left,
  minor y tick num=4,
  xlabel=$\varepsilon$,
  ylabel={Space saving (\%)},
  %ylabel shift=-7pt,
  grid=both,
  grid style={dotted},
  legend entries={{Web logs},{Longitude},{IoT}},
  legend cell align={left},
  legend style={
    at={(0.987,0.98)},
    anchor=north east,
    draw=white!80.0!black,
    /tikz/every even column/.append style={column sep=3mm},
    nodes={scale=0.7, transform shape}
  },
  mark options={solid},
  every axis plot/.append style={fill,draw=black},
  cycle list name=printerfriendly,
  legend image code/.code={
    \draw [fill=#1,draw=black] (0cm,-0.05cm) rectangle (0.4cm,0.08cm);
  },
  preaction={draw=black!95, line width=0.75pt},
  set layers=mylayers,
]

\addplot
table {%
8		81.2
16		81.1
32		79.2
64		78.4
128		76.6
256		72.9
512     65.4
1024	53.8
2048	37.9
};

\addplot
table {%
8		80.5
16		76.2
32		68.7
64		53.8
128		32.3
256		11.1
512		-5.9
1024	-11.6
2048	-14.3
};

\addplot 
table {%
8		79.7
16		77.9
32		74.5
64		70.4
128		63.5
256		57.2
512		52.4
1024	45.4
2048	40.4
};

\end{axis}
\end{tikzpicture}
  \vspace*{-0.5\baselineskip}
  \caption{Slope compression reduces the {\em space} taken by the slopes up to 81.2\%. Longitude is the only dataset on which compression does not help for $\varepsilon\geq2^9$ because of its special features. In this case compression would not be adopted.}
  \label{fig:slope-space-savings}
\end{figure}
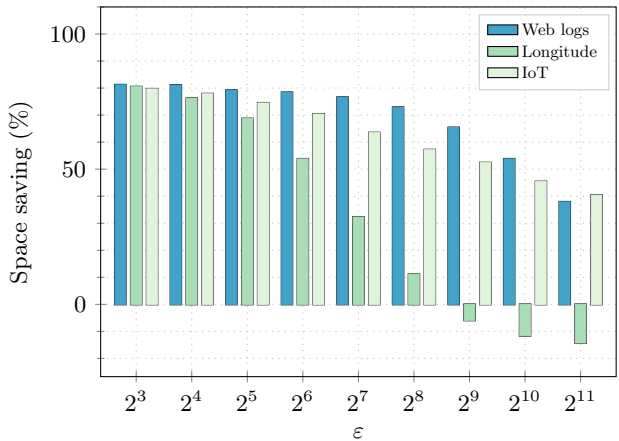

In information theory, the compressibility of data is measured with its entropy (as defined by Shannon). We conjecture that a similar measure, characterising ``difficult'' datasets like \emph{Longitude}, also exists in our ``geometric setting'' and is worth studying. Another interesting future work is to explore the relation between the algorithm of \cref{lem:optimal-pla-model} and the slope compression algorithm. To explain, recall that during the construction of the optimal \gls{plamodel} the range of slopes reduces each time a new point is added to the current convex hull (segment). Therefore ``shortening'' a segment, on the one hand, improves the performance of the slope compression algorithm (because it enlarges the sizes of the possible slope intervals), but on the other hand, it increases the overall number of segments. Given that \cref{lem:slope-compression-algorithm} offers a compressed space bound which depends on $m$ (the overall number of segments) and $t$ (the number of distinct segments), balancing the above two effects to achieve better compression is an intriguing extension of this paper.

Afterwards, we measured the query performance of the Compressed \pgmindex in which compression was activated over the intercepts and the slopes of the segments of all the levels. \cref{tab:compressed} shows that, with respect to the corresponding Recursive \pgmindex, the space footprint is reduced by up to 52.2\% at the cost of moderately slower queries (no more than 24.5\%).

\begin{table}
  \setlength{\tabcolsep}{0.5em}
  \centering
  \begin{tabular}{r r r r r r r}
    \toprule
    $\varepsilon=$    &   64 &  128 &  256 &  512 & 1024 & 2048 \\
    \midrule
    Space saving (\%) & 52.2 & 50.8 & 48.5 & 46.0 & 41.5 & 35.5 \\
    Time loss (\%)    & 13.7 & 22.6 & 24.5 & 15.1 & 11.7 &  9.9 \\
    \bottomrule
  \end{tabular}%\bigskip
  \caption{Query performance of the Compressed \pgmindex with respect to the corresponding Recursive \pgmindex.}
  \label{tab:compressed}
\end{table}

\subsection{The Multicriteria \pgmindex}\label{ssec:exp-multicriteria-pgm-index}

Our implementation of the Multicriteria \pgmindex operates in two modes: the time-minimisation mode (shortly, min-time) and the space-minimisation mode (shortly, min-space), which implement the algorithms described in \cref{sec:multicriteria-pgm-index}. In min-time mode, inputs to the program are $s_{max}$ and a tolerance \id{tol} on the space occupancy of the solution, and the output is the value of the error $\varepsilon$ which guarantees a space bound $s_{max} \pm \id{tol}$. In min-space mode, inputs to the program are $t_{max}$ and a tolerance \id{tol} on the query time of the solution, and the output is the value of the error $\varepsilon$ which guarantees a time bound $t_{max} \pm \id{tol}$ in the query operations. We note that the introduction of a tolerance parameter allows us to stop the search earlier as soon as any further step would not appreciably improve the solution (i.e., we seek only improvements of several bytes or nanoseconds). So \id{tol} is not a parameter that has to be tuned but rather a stopping criterion like the ones used in iterative methods.

To model the space occupancy of a \pgmindex, we studied empirically the behaviour of the number of segments $m_{opt}=s_L(\varepsilon)$ forming the optimal \gls{plamodel}, by varying $\varepsilon$ and by fitting ninety different functions over about two hundred points $(\varepsilon, s_L(\varepsilon))$ generated beforehand by a long-running grid search over our real-world datasets. Looking at the fittings, we chose to model $s_L(\varepsilon)$ with a power-law having the form $a \varepsilon^{-b}$. As further design choices we point out that: (i)~the fitting of the power-law was performed with the Levenberg-Marquardt algorithm, while root-finding was performed with Newton's method; (ii)~the search space for $\varepsilon$ was set to $\mathcal E = [8,n/2]$ (since a cache line holds eight 64 bits integers); and finally (iii)~the number of guesses was set to $2\lceil{\log\log\mathcal E}\rceil$.

\smallskip The following experiments were executed by addressing some use cases in order to show the efficacy and efficiency of the multicriteria \pgmindex.

\myparagraph{Experiments with the min-time mode}
Suppose that a database administrator wants the most efficient \pgmindex for the \emph{Web logs} dataset that fits into an L2 cache of 1 MiB. Our solver derived an optimal \pgmindex matching that space bound by setting $\varepsilon=393$ and taking 10 iterations and a total of 19 seconds. This result was obtained by approximating $s_L(\varepsilon)$ with the power-law $46032135 \cdot \varepsilon^{-1.16}$ which guaranteed a mean squared error of no more than 4.8\% over the range $\varepsilon\in[8,1024]$.

As another example, suppose that a database administrator wants the most efficient \pgmindex for the \emph{Longitude} dataset that fits into an L1 cache of 32 KiB. Our solver derived an optimal \pgmindex matching that space bound by setting $\varepsilon=1050$ and taking 14 iterations and a total of 9 seconds.

\myparagraph{Experiments with the min-space mode} Suppose that a database administrator wants the \pgmindex for the IoT dataset with the lowest space footprint that answers any query in less than 500 ns. Our solver derived an optimal \pgmindex matching that time bound by setting  $\varepsilon=432$, occupying 74.55 KiB of space, and taking 9 iterations and a total of 6 seconds.

As another example, suppose that a database administrator wants the most compressed \pgmindex for the \emph{Web logs} dataset that answers any query in less than 800 ns. Our solver derived an optimal \pgmindex matching that time bound by setting  $\varepsilon=1217$, occupying 280.05 KiB of space, and taking 8 iterations and a total of 17 seconds.

\myparagraph{Discussion} In contrast to the \fitingtree and the \gls{rmi}, the Multicriteria \pgmindex can trade efficiently query time with space occupancy, making it a promising approach for applications with rapidly-changing data distributions and space/time constraints. Overall, in both modes our approach ran in less than 20 seconds.

\section{Conclusions and future work}
\label{sec:conclusions}

We have introduced the \pgmindex, a learned data structure for the fully indexable dictionary problem which improves the query performance and the space occupancy of both traditional and modern learned indexes up to several orders of magnitude. We have also designed three variants of the \pgmindex: one that improves its already succinct space footprint using ad-hoc compression techniques, one that adapts itself to the query distribution, and one that provides estimations of the number of occurrences of individual or range of keys. Finally, we have demonstrated that the \pgmindex is a multicriteria data structure as it can be fast optimised within some user-given constraint on the space occupancy or the query time.

A possible research direction is to experiment with the performance of insertion and deletion of keys in a \pgmindex. To this end, we mention classic techniques such as the split-merge strategy in \btree nodes~\cite{Bender:2005,Vitter:2001}, and the use of buffers that once full are merged into the index (cf. \cite{Galakatos:2019}). The possibility of orchestrating segments, nonlinear models and rank/select indexing techniques from the compression domain~\cite{Navarro:2016,Witten:1999} is another intriguing research direction, especially within our Multicriteria framework. 

\section{Acknowledgments}

Part of this work has been supported by the Italian MIUR PRIN project ``Multicriteria Data Structures and Algorithms: from compressed to learned indexes, and beyond'' (Prot. 2017WR7SHH) and by Regione Toscana (under POR FSE 2014/2020).

\bibliographystyle{abbrv}
\bibliography{bibliography}

\begin{thebibliography}{10}

\bibitem{Bagchi:2005}
A.~Bagchi, A.~L. Buchsbaum, and M.~T. Goodrich.
\newblock Biased skip lists.
\newblock {\em Algorithmica}, 42(1):31--48, 2005.

\bibitem{Bender:2005}
M.~Bender, E.~Demaine, and M.~Farach-Colton.
\newblock Cache-oblivious b-trees.
\newblock {\em {SIAM} J. Comput.}, 35(2):341--358, 2005.

\bibitem{Bent:1985}
S.~W. Bent, D.~D. Sleator, and R.~E. Tarjan.
\newblock Biased search trees.
\newblock {\em {SIAM} J. Comput.}, 14(3):545--568, 1985.

\bibitem{Bingmann:2013}
T.~Bingmann.
\newblock {STX} {B}\textsuperscript{+}-tree {C}\texttt{++} template classes,
  2013.
\newblock \url{http://panthema.net/2007/stx-btree}. Version 0.9.

\bibitem{Buragohain:2007}
C.~Buragohain, N.~Shrivastava, and S.~Suri.
\newblock Space efficient streaming algorithms for the maximum error histogram.
\newblock In {\em Proceedings of the IEEE 23rd International Conference on Data
  Engineering}, ICDE, pages 1026--1035, Washington, D.C., USA, 2007. {IEEE}
  Computer Society.

\bibitem{Burtscher:2009}
M.~{Burtscher} and P.~{Ratanaworabhan}.
\newblock Fpc: A high-speed compressor for double-precision floating-point
  data.
\newblock {\em IEEE Transactions on Computers}, 58(1):18--31, 2009.

\bibitem{Chan:1998}
C.-Y. Chan and Y.~E. Ioannidis.
\newblock Bitmap index design and evaluation.
\newblock In {\em Proceedings of the ACM International Conference on Management
  of Data}, SIGMOD, pages 355--366, New York, NY, USA, 1998. ACM.

\bibitem{Chen:2013}
D.~Z. Chen and H.~Wang.
\newblock Approximating points by a piecewise linear function.
\newblock {\em Algorithmica}, 66(3):682--713, 2013.

\bibitem{Chen:2007}
Q.~Chen, L.~Chen, X.~Lian, Y.~Liu, and J.~X. Yu.
\newblock Indexable pla for efficient similarity search.
\newblock In {\em Proceedings of the 33rd International Conference on Very
  Large Data Bases}, VLDB, pages 435--446. VLDB Endowment, 2007.

\bibitem{Cormode:2017}
G.~Cormode.
\newblock Data sketching.
\newblock {\em Communications of the {ACM}}, 60(9):48--55, 2017.

\bibitem{Farruggia:2014}
A.~Farruggia, P.~Ferragina, A.~Frangioni, and R.~Venturini.
\newblock Bicriteria data compression.
\newblock In {\em Proceedings of the 25th ACM-SIAM Symposium on Discrete
  Algorithms}, SODA, pages 1582--1595, Philadelphia, PA, USA, 2014. Society for
  Industrial and Applied Mathematics.

\bibitem{Ferragina:2016}
P.~Ferragina and R.~Venturini.
\newblock Compressed cache-oblivious {S}tring {B}-tree.
\newblock {\em ACM Trans. Algorithms}, 12(4):52:1--52:17, 2016.

\bibitem{Galakatos:2019}
A.~Galakatos, M.~Markovitch, C.~Binnig, R.~Fonseca, and T.~Kraska.
\newblock Fiting-tree: A data-aware index structure.
\newblock In {\em Proceedings of the 2019 International Conference on
  Management of Data}, SIGMOD, pages 1189--1206, New York, NY, USA, 2019. ACM.

\bibitem{Graefe:2006}
G.~Graefe.
\newblock B-tree indexes, interpolation search, and skew.
\newblock In {\em Proceedings of the 2Nd International Workshop on Data
  Management on New Hardware}, DaMoN, New York, NY, USA, 2006. ACM.

\bibitem{Grama:2003}
A.~Grama, G.~Karypis, V.~Kumar, and A.~Gupta.
\newblock {\em Introduction to Parallel Computing}.
\newblock Addison-Wesley Longman Publishing Co., Inc., Boston, MA, USA, 2
  edition, 2003.

\bibitem{Idreos:2019}
S.~Idreos, K.~Zoumpatianos, S.~Chatterjee, W.~Qin, A.~Wasay, B.~Hentschel,
  M.~Kester, N.~Dayan, D.~Guo, M.~Kang, and Y.~Sun.
\newblock Learning data structure alchemy.
\newblock {\em Bulletin of the IEEE Computer Society Technical Committee on
  Data Engineering}, 42(2):46--57, 2019.

\bibitem{Kraska:2019}
T.~Kraska, M.~Alizadeh, A.~Beutel, E.~H. Chi, A.~Kristo, G.~Leclerc, S.~Madden,
  H.~Mao, and V.~Nathan.
\newblock Sagedb: {A} learned database system.
\newblock In {\em Proceedings of the 9th Biennial Conference on Innovative Data
  Systems Research}, CIDR, 2019.

\bibitem{Kraska:2018}
T.~Kraska, A.~Beutel, E.~H. Chi, J.~Dean, and N.~Polyzotis.
\newblock The case for learned index structures.
\newblock In {\em Proceedings of the 2018 International Conference on
  Management of Data}, SIGMOD, pages 489--504, New York, NY, USA, 2018. ACM.

\bibitem{Lindstrom:2018}
P.~Lindstrom, S.~Lloyd, and J.~Hittinger.
\newblock Universal coding of the reals: Alternatives to ieee floating point.
\newblock In {\em Proceedings of the Conference for Next Generation
  Arithmetic}, CoNGA, pages 5:1--5:14, New York, NY, USA, 2018. ACM.

\bibitem{Mehlhorn:1984}
K.~Mehlhorn.
\newblock {\em Data Structures and Algorithms 1: Sorting and Searching},
  volume~1 of {\em {EATCS} Monographs on Theoretical Computer Science}.
\newblock Springer-Verlag, Berlin, Heidelberg, New York, Tokyo, 1984.

\bibitem{Moffat:2002book}
A.~Moffat and A.~Turpin.
\newblock {\em Compression and Coding Algorithms}.
\newblock Springer, Boston, MA, USA, 2002.

\bibitem{Naono:2010}
K.~Naono, K.~Teranishi, J.~Cavazos, and R.~Suda, editors.
\newblock {\em Software Automatic Tuning, From Concepts to State-of-the-Art
  Results}.
\newblock Springer, New York, Dordrecht, Heidelberg, London, 2010.

\bibitem{Navarro:2016}
G.~Navarro.
\newblock {\em Compact data structures: A practical approach}.
\newblock Cambridge University Press, New York, NY, USA, 2016.

\bibitem{Navarro:2007}
G.~Navarro and V.~M\"{a}kinen.
\newblock Compressed full-text indexes.
\newblock {\em {ACM} Comput. Surv.}, 39(1):2, 2007.

\bibitem{Okanohara:2007}
D.~Okanohara and K.~Sadakane.
\newblock Practical entropy-compressed rank/select dictionary.
\newblock In {\em Proceedings of the Meeting on Algorithm Engineering \&
  Expermiments}, pages 60--70, Philadelphia, PA, USA, 2007. Society for
  Industrial and Applied Mathematics.

\bibitem{OpenStreetMap}
{OpenStreetMap contributors}.
\newblock {OpenStreetMap Data Extract for Italy}.
\newblock \url{https://www.openstreetmap.org}, 2018.
\newblock Retrieved from \url{http://download.geofabrik.de} on May 9, 2018.

\bibitem{ORourke:1981}
J.~O'Rourke.
\newblock An on-line algorithm for fitting straight lines between data ranges.
\newblock {\em Commun. ACM}, 24(9):574--578, 1981.

\bibitem{Pagh:2004}
R.~Pagh and F.~F. Rodler.
\newblock Cuckoo hashing.
\newblock {\em Journal of Algorithms}, 51(2):122 -- 144, 2004.

\bibitem{Petrov:2018}
A.~Petrov.
\newblock Algorithms behind modern storage systems.
\newblock {\em Commun. ACM}, 61(8):38--44, 2018.

\bibitem{Patrascu:2006}
M.~P\u{a}tra\c{s}cu and M.~Thorup.
\newblock Time-space trade-offs for predecessor search.
\newblock In {\em Proceedings of the 38th ACM Symposium on Theory of
  Computing}, STOC, pages 232--240, New York, NY, USA, 2006. ACM.

\bibitem{Rao:1999}
J.~Rao and K.~A. Ross.
\newblock Cache conscious indexing for decision-support in main memory.
\newblock In {\em Proceedings of the 25th International Conference on Very
  Large Data Bases}, VLDB, pages 78--89, San Francisco, CA, USA, 1999. Morgan
  Kaufmann Publishers Inc.

\bibitem{Seidel:1996}
R.~Seidel and C.~R. Aragon.
\newblock Randomized search trees.
\newblock {\em Algorithmica}, 16(4/5):464--497, 1996.

\bibitem{Sklansky:1980}
J.~Sklansky and V.~Gonzalez.
\newblock Fast polygonal approximation of digitized curves.
\newblock {\em Pattern Recognition}, 12(5):327 -- 331, 1980.

\bibitem{Vitter:2001}
J.~S. Vitter.
\newblock External memory algorithms and data structures: Dealing with massive
  data.
\newblock {\em ACM Comput. Surv.}, 33(2):209--271, 2001.

\bibitem{Wang:2017}
J.~Wang, C.~Lin, Y.~Papakonstantinou, and S.~Swanson.
\newblock An experimental study of bitmap compression vs. inverted list
  compression.
\newblock In {\em Proceedings of the 2017 ACM International Conference on
  Management of Data}, SIGMOD, pages 993--1008, New York, NY, USA, 2017. ACM.

\bibitem{Witten:1999}
I.~H. Witten, A.~Moffat, and T.~C. Bell.
\newblock {\em Managing Gigabytes (2Nd Ed.): Compressing and Indexing Documents
  and Images}.
\newblock Morgan Kaufmann Publishers Inc., San Francisco, CA, USA, 1999.

\bibitem{Xie:2014}
Q.~Xie, C.~Pang, X.~Zhou, X.~Zhang, and K.~Deng.
\newblock Maximum error-bounded piecewise linear representation for online
  stream approximation.
\newblock {\em The VLDB Journal}, 23(6):915--937, 2014.

\end{thebibliography}

\end{document}